\documentclass[12pt]{article}

\usepackage{geometry}
\usepackage[utf8]{inputenc}
\usepackage[T1]{fontenc}
\usepackage[tt=false,type1=true]{libertine}
\usepackage[varqu]{zi4}
\usepackage[libertine]{newtxmath}

\usepackage{hyperref}       
\usepackage{url}            
\usepackage{booktabs}       
\usepackage{amsfonts}       
\usepackage{nicefrac}       
\usepackage{microtype}      
\usepackage{xcolor}         

\usepackage[ruled]{algorithm2e} 

\SetAlFnt{\small}
\SetAlCapFnt{\small}
\SetAlCapNameFnt{\small}
\SetAlCapHSkip{0pt}
\IncMargin{-\parindent}

\usepackage{tikz}
\usepackage{slashbox}
\usepackage{caption}
\usepackage{subcaption}
\usepackage{mathrsfs}
\usepackage{amssymb, amsmath, latexsym, amsfonts, graphicx, amsthm}
\usepackage{natbib}

\setcitestyle{authoryear}

\DeclareMathOperator*{\E}{\mathbb{E}}

\newcommand{\Rev}{\textsc{Rev}}

\newcommand{\Imply}{\Longrightarrow}

\newcommand{\F}{\mathcal{F}}
\newcommand{\G}{\mathcal{G}}
\newcommand{\x}{\tilde{x}}
\newcommand{\one}{\mathbf{1}}

\providecommand{\U}[1]{\protect\rule{.1in}{.1in}}
\newtheorem{theorem}{Theorem}[section]
\newtheorem{corollary}[theorem]{Corollary}
\newtheorem{lemma}[theorem]{Lemma}
\newtheorem{proposition}[theorem]{Proposition}

\newtheorem*{lemma*}{Lemma}

\title{Calibrated Click-Through Auctions:\\ An Information Design Approach}

\author{Dirk Bergemann\thanks{Yale University, dirk.bergemann@yale.edu} \and
Paul Duetting\thanks{Google Research, duetting@google.com} \and
Renato Paes Leme\thanks{Google Research, renatoppl@google.com} \and
Song Zuo\thanks{Google Research, szuo@google.com}}

\begin{document}

\maketitle

\begin{abstract}
We analyze the optimal information design in a click-through auction with fixed valuations per click, but stochastic click-through rates. While the auctioneer takes as given the auction rule of the click-through auction, namely the generalized second-price auction, the auctioneer can design the information flow regarding the click-through rates among the bidders.
A natural requirement in this context is to ask for the information structure to be calibrated in the learning sense.
With this constraint, the auction needs to rank the ads by a product of the bid and an \emph{unbiased} estimator of the click-through rates, and the task of designing an optimal information structure is thus reduced to the task of designing an optimal unbiased estimator.

We show that in a symmetric setting with uncertainty about the click-through rates, the optimal information structure attains both social efficiency and surplus extraction. 
The optimal information structure requires private (rather than public) signals to the bidders. It also requires correlated (rather than independent) signals, even when the underlying uncertainty regarding the click-through rates is independent.
Beyond symmetric settings, we show that the optimal information structure requires partial information disclosure.



\end{abstract}

\thispagestyle{empty}
\setcounter{page}{0}
\clearpage





\section{Introduction}

In the world of digital advertising, whether in display or in
search advertising, the allocation mechanism is commonly an auction. Independent of the details of the auction format, the mechanism
typically elicits from each bidder the willingness-to-pay for the item.
Importantly, the auction frequently uses additional information to determine
the ranking in a search auction or the assignment in the display
advertisement. The additional information frequently concerns the quality of
the match between the bidder (the advertiser) and the item (attention of the
consumer). Importantly, this additional information, such as the expected
click-through rate or expected transaction rate, is often held by the
platform that manages the auction, or the publisher on whose site the
consumer is present. Consequently, the relevant information for the
auction comes from many information sources, some of it provided by the
advertiser (the bidder), some of it provided by the publisher (the seller), or the auction platform.

As the relevant information arises from many sources, the allocation
mechanism must determine how much information to elicit and how to integrate
this information into the allocation mechanism. We pursue this question in a
simple, yet prominent setting---namely, the click-through auction as analyzed in \citet
{edos07} and \citet{vari07}, who refer to this type of auction as the
\emph{generalized second-price auction}, or \emph{position auction},
respectively. In this environment, the utility of a bidder is given by the
product of the willingness-to-pay of the bidder per click of the consumer,
and the click-through rate of the consumer. In the following, we shall simply
refer to this auction mechanism as the \emph{click-through auction}.

\citet{edos07} and \citet{vari07} consider a complete information setting
where the willingness-to-pay of each bidder and the click-through rate for
each position are assumed to be known by all auction participants. They
showed that there is an equilibrium in the bidding game of the click-through
auction in which the bidder with the highest product of the willingness-to-pay and the click-through rate wins. Thus, the click-through auction
supports the socially efficient outcome, and if the auction offers several
positions or rankings, then the socially efficient outcome extends to all
offered positions. Significantly, in equilibrium the price for each position
reflects the marginal contribution of each bidder to the social surplus.
Thus, the resulting payoffs of the bidders and the auctioneer are as they
would be in the direct Vickrey-Groves-Clarke (VCG)\ mechanism. The payoff
equivalence of the click-through auction with the corresponding VCG
mechanism is critical as it suggests that the mechanism is robust to the
introduction of private information about either the willingness-to-pay or
the click-through rate.

In the subsequent analysis, we take as given the
click-through auction format \emph{and} the profile of willingness-to-pay across
bidders. In contrast to the above contributions, we allow for randomness
and uncertainty in the click-through rates. The publisher or auction
platform must then decide how much of the information to share with the
auction participants. For the subsequent analysis, we do not distinguish
between the publisher and the auction platform, and simply refer to either
(or both of them) as the seller. The central instrument, and thus the focus
of our analysis, is the \emph{information policy} of the seller with respect
to the true click-through rates generated by any given search event.

We approach the problem of determining the optimal information policy in a
number of steps in increasing generality. We focus on randomness in the
click-through rates and maintain complete information about the willingness-to-pay of each bidder throughout the analysis. Each search event generates a
random click-through rate for each of the participating bidders. The search
environment is therefore described by a joint distribution of the
click-through rates across the bidders. The joint distribution of the
click-through rates is common information to all the participants in the
auction, bidders and seller, and is thus given by a common prior distribution. The
initial information of the bidders is given by the common prior
distribution. In contrast, the seller is assumed to know the realized
click-through rates of each search event. The information policy of the
seller then has to determine how much information to disclose about the
realized click-through rates before bidding begins.

With the randomness of the click-through rates it is natural to constrain
the information design of the platform to be consistent with the prior
distribution. We therefore refer to an information policy that maintains the
law of iterated expectation as \emph{calibrated} {(in the sense of \citet{fovo97})}. Among all possible information policies, the complete
information policy and the zero information policy are both leading examples,
as well as extremal information policies. Under a complete information policy,
the seller completely discloses all information to the bidders. It is thus
as if the bidders were in a complete information environment. By contrast, in
a zero information policy the seller does not disclose any information
about the realized click-through rates. In consequence, each bidder acts as if
the realized click-through rate is always equal to the ex ante expected
click-through rate. These two extremal information policies have
dramatically different payoff implications.

With stochastic click-through rates, social efficiency and revenue will depend on which information about the click-through rates is disclosed.
With complete information, the resulting allocation is always efficient. 
But as the competitive position of each
bidder can vary across the realized click-through rates, the resulting
revenue of the seller can be low due to weak competition. By contrast,
with zero information about the click-through rates, the resulting allocation
will typically fail to be socially efficient. As the bidding behavior cannot
reflect any information about the click-through rates, the socially relevant
information fails to be reflected in the auction outcome. Yet, as the
bidding reflects only the expected click-through rate, and hence only the mean of
the click-through rate (and not the higher moments of the click-through
distribution), the resulting bids will have zero variance, and thus be more
competitive.

Our first result, Proposition \ref{ics}, shows that a statistically independent information structure can never improve the revenue over the no-disclosure information policy. Our second result, Corollary \ref{zc}, shows that with any level of uncertainty in the click-through rates, the seller strictly prefers no-disclosure to full-disclosure. Thus, the seller always favors competition over information disclosure. We then
ask whether there exists an improved information policy that can realign social
efficiency with the revenue of the seller.

Our main result, Theorem \ref{cs}, establishes
that a calibrated and correlated information policy can completely align
social efficiency and revenue maximization. In particular, when the common
prior distribution over the click-through rates is symmetric across the
bidders, we can then explicitly construct an information policy such that
the socially efficient allocation is always realized and the
competition between the bidder levels the information rent of the
bidders to the ex-ante level. Given the symmetry of the common prior
distribution, this implies that the bidders compete their residual surplus
down to zero.

We also provide an explicit construction of the correlated information structure
(or signalling scheme). Interestingly, the optimal information structure is
an interior information structure; that is, it is neither zero nor complete
information disclosure. The information structure balances two conditions
that are necessary to attain the socially efficient allocation while
maintaining competition:\ (i) it provides sufficient information to rank the
alternative allocations according to social efficiency, and (ii) it limits
the variance in the posterior beliefs of the competing bidders so that their
equilibrium bids remain arbitrarily close to support competitive bids.

The optimality of a noisy information structure has some significant
implications in the world of digital advertising. Since the optimal
information structure remains noisy, better click-through-rate
predictions, achieved through improved learning, may not necessarily lead to better
auction results. Thus, there might be limits to the returns of more elaborate
machine learning algorithms to inform the prediction problem. Importantly,
the calibration constraint makes sure that the ranking score always remains
an unbiased click-through-rate predictor.

The optimality of an interior information structure remains even when we
move away from the symmetric stochastic environment. In Proposition \ref{equal-means}, we show
that the socially efficient allocation and revenue maximization remain
perfectly aligned as long as the expected click-through rate is equalized
across bidders, even when the support of the ex-post realized
click-through rates can vary across bidders. Finally, in Theorem \ref{thm:moderate},
we show that an interior information structure remains part of the optimal
information design with stochastic click-through rates, even when the expected
click-through rates across the bidders differ, and therefore one bidder is
stronger from an ex-ante perspective. In particular, the optimal information
structure releases less information about the winning bidder than about the
losing bidder. This suggests that the optimal information structure in an
asymmetric auction seeks to strengthen the weak bidder with additional
information, relative to the strong bidder.

\subsection{Related Literature}

\citet{bepe07} and \citet{essz07} are among the first to investigate the
design of optimal information structure in an auction setting. \citet{bepe07} consider an auction environment with $n$ bidders and independent
private values. In their setting, the bidders initially have no private
information, and the platform can distribute any information. Given the
independent value assumption, \citet{bepe07} restrict the auctioneer to
transmit information regarding bidder $i$'s valuation only to bidder $i$.
Thus, there is a restriction in the signalling scheme. By contrast, the
seller can adapt the selling mechanism to the information structure, and is
thus not restricted to the second-price auction. Consequently, the seller jointly
optimizes auction and information policies. \citet{essz07} extend the
analysis of \citet{bepe07} to permit the bidder to initially have private
information, and then allow the seller to propose the optimal mechanism.

More recently, a number of contributions have analyzed the optimal
information structure for a given mechanism or auction format. \citet{bebm15}
consider the optimal information structure for a seller who uses
third-degree price discrimination, and \citet{bebm17} pursue this question in
the context of the first-price auction. \citet{babx19} consider a second-price auction where the information about the valuation of bidder $i$ is
partially shared between the bidder and the platform. The limiting
cases are that: $(i)$ the bidder has all the private
information (the standard setting of auction theory), and $(ii)$ the platform has all the private information.

Considering the randomness of the click-through rates, earlier literature
notices that a modified version of the click-through rates may positively
impact the revenue of the click-through auction. In this context, the
adjusted click-through rates in \citet{lape07} and \citet{masu15} are not
calibrated, as the true click-through rate is "squashed" or "boosted" for the
purpose of ranking. This suggests a natural distinction between the analysis
of optimal calibrated and non-calibrated information structures.

The design of optimal information structures in a strategic setting remains an
area with many open questions. In the current setting, we explicitly allow
for private information disclosure, rather than public information
disclosure. By contrast, in most of the preceding literature the information
disclosure was either public, as in \citet{arba19}, or independent across bidders, as in \citet{bepe07}. Here, we are allowing for, and importantly showing, the optimality of private and correlated information structures.

In the current setting, the auctioneer can substantially improve the revenue by using correlated instead of independent signals. The significance of correlation among the signals for the revenue maximizing mechanism has been observed earlier in \citet{crmc85,crmc88}. They establish that correlation in the private values among bidders can be used in an optimal mechanism to extract the full surplus. While our results also highlight the increased power of correlated signals relative to independent signals, the setting and arguments differ substantially. In \citet{crmc85,crmc88}, the auctioneer is free to choose the optimal mechanism, while we take the generalized second-price auction as given. In \citet{crmc85,crmc88}, the signal of each bidder is the private value of the bidder, and the values need to be correlated for the full surplus extraction result to be obtained. In our setting, the values of the bidders are known, and the click-through rates themselves can either be independent or correlated. We choose the information structure so that the signals are sufficiently informative, yet yield competitive interim expectations. By contrast, in \citet{crmc85,crmc88}, the payments of each bidder depend on the reports of the other bidders, so that the individual reporting strategy seeks to maximize a scoring function. Thus, the correlation of the signals in these two settings seeks to achieve very different objectives, and accordingly the construction differs significantly. In particular, \citet{crmc85,crmc88} take as given the signals and then design the optimal transfer function. We take as given the transfer function, namely the payment rules, and then design the signals to maintain competition.
\footnote{We should also note that correlation is not always part of an optimal information structure. In the case of weak competition, that is, when the efficient allocation always assigns the object to the same bidder, then the optimal information structure will not use correlated signals.}

\section{Model}

We will analyze the setting of click-through auctions where bidders are ranked by the product of their value (expressed as a maximum willingness-to-pay per click) and an unbiased estimator of the click-through rates. Our main goal will be to study how to engineer such unbiased estimators to achieve better revenue-efficiency trade-offs. In the language of information design, we will keep the auction format fixed and vary the information structure.

To allow us to focus on the information structure, we will work in the full information model where each bidder $i=1,\ldots,n$ has a fixed and known value $v_i \geq 0$ representing their willingness-to-pay for a click. Our central object of study will be the click-through rates (CTRs): before the auction, a vector $r=(
r_{1},\ldots ,r_{n}) \in [0,1] ^{n}$ will be drawn from a joint prior distribution $G$, which is a multi-dimensional distribution that may display correlation across the bidders' CTRs.


The click-through rates are known by the auctioneer: typically, the platform is the one building a machine learning model to estimate them. The auctioneer must now decide on a score/signal $s = (s_1,  \ldots, s_n) \in [0, 1]^n$ to rank the bidders. The design space will be to design a joint probability distribution $\rho$ on pairs $(r,s)$ such that the marginal on $r$ is $G$:
$$\int_{r \in R} dG(r) = \int_{r \in R}  \int_{s} d\rho(s,r), \quad \forall \text{ measurable } R \subseteq [0,1]^n.$$
This joint probability distribution will be referred as the \emph{information structure}. For notational convenience, it will be useful to assume that both $G$ and $\rho$ are discrete distributions with finite support, and hence we can write $g(r)$ for the probability of a given vector $r$ under distribution $G$ and $x(r,s)$ as the probability of a pair $(r,s)$ under distribution $\rho$. With that notation, the information structure is a function $x : [0,1]^n \times [0,1]^n \rightarrow [0,1]$ satisfying:
$$\sum_s x(r,s) = g(r), \quad \forall r \in [0, 1]^n.$$

\paragraph{Auction Mechanics} Again for simplicity, we focus on the single slot setting where the goal of the auction is to select a single winner $i^\ast \in [n]$.
The winner will be selected as the bidder having the largest $s_i v_i$, with a symmetric tie-breaking rule. The winner's cost per click is then:
\begin{equation*}
p_{{i}^{\ast }}=\max_{j\neq i^{\ast }}\frac{v_{j}s_{j}}{s_{i^{\ast }}}.
\end{equation*}
The winner only pays when there is a click, which happens with probability $r_{i^\ast}$. Hence, the expected revenue from this auction is $r_{i^\ast} p_{i^\ast}$.

\paragraph{Calibration} The auctioneer is restricted to ranking with the unbiased estimator of the click-through rates. Hence, the information structure will be valid only if it is calibrated in the Foster-Vohra sense \citep{fovo97}. An information structure is called \emph{calibrated} if the posterior, given
any signal realization $s_{i}^{\prime }$ for bidder $i$, matches with the
signal itself, i.e.,
\begin{equation}
\mathbb{E}[r_{i}|s_{i}=s_{i}^{\prime }]=s_{i}^{\prime }.  \label{cal}
\end{equation}%
If the CTR and signal space is discrete, then we can write calibration as:
$$\sum_{(r,s); s_i = s'_i} x(r,s) \cdot (r_i-s'_i) = 0, \quad \forall i, s'_i.$$

There are two important examples of calibrated information structures: 
\begin{itemize}
    \item Full-disclosure: where $s_i = r_i$ almost surely. 
    \item No-disclosure: where $s_i = \E[r_i]$ almost surely. 
\end{itemize}
Since the calibration constraint is imposed on every bidder separately, it is possible to create information structures that combine disclosure and no-disclosure. For example, given two bidders, we can consider an information structure where bidder $1$ receives only one signal, and the signal is equal to the ex-ante expectation of the click-through rate, thus  $s_1 = \E[r_1]$; and bidder $2$ receives as many signals as click-through rates, thus $s_2 = r_2$. This forms a calibrated information structure.

Whenever $s_i = \E[r_i]$ we will say that we \emph{fully bundle} bidder $i$. Whenever $s_i = r_i$ we will say that we \emph{unbundle} bidder $i$. If neither is the case we will say that we \emph{partially bundle} the bidder.

\paragraph{Independence and Correlation}

A information structure is called \emph{independent} if signals $s_j$ for $j \neq i$ do not offer additional information on the expectation of $r_i$ beyond $s_i$. Formally:
\begin{equation}
\E[r_i | s = s'] = \E[r_i | s_i = s'_i],
\ \ \ \forall i,\forall s\text{. }
\label{in}%
\end{equation}

Both full-disclosure and no-disclosure information structures are independent.

Whenever we do not assume independence, we will say that an information structure is \emph{correlated}. Below, we give an example of a correlated and calibrated information structure. This will also serve as an example of how information structures will be illustrated throughout the paper. Consider the two-bidder setting where CTRs are $r \in \{(1/2,1/2), (1/2,1), (1,1/2), (1,1)\}$, each with  probability $1/4$ (hence $r_1$ and $r_2$ are independent). We represent an information structure where rows correspond to pairs of CTRs $(r_1, r_2)$ and columns correspond to pairs of signals $(s_1, s_2)$. Each entry of the matrix will correspond to $x(r,s)$ which is the probability of the event that the CTRs are $r$ and the signals are $s$. The following table shows the "\emph{flipping the square}`` structure, which we will discuss in detail in Section \ref{sec:square}:

\medskip

\begin{center}
\begin{tabular}
[c]{l|c|c|c|c}%
\toprule
\backslashbox{$r$}{$s$} & $(3/4-\epsilon,3/4-\epsilon)$ &
$(3/4-\epsilon,3/4+\epsilon)$ & $(3/4+\epsilon,3/4-\epsilon)$ &
$(3/4+\epsilon,3/4+\epsilon)$\\\hline
$\left(  1/2,1/2\right)  $ & $\epsilon$ & $0$ & $0$ & $1/4-\epsilon
$\\
$(1/2,1)$ & $0$ & $1/4$ & $0$ & $0$\\
$(1,1/2)$ & $0$ & $0$ & $1/4$ & $0$\\
$(1,1)$ & $1/4-\epsilon$ & $0$ & $0$ & $\epsilon$\\
\bottomrule
\end{tabular}
\end{center}

\medskip

One can check that while the calibration constraints (equation \eqref{cal}) hold, the independence condition (equation \eqref{in}) does not. So, this is a calibrated, correlated information structure.





\paragraph{Symmetric vs Asymmetric Environments} Finally, we will say that the environment is \emph{symmetric} when random variables $v_1 r_1,\hdots, v_n r_n$ are exchangeable.
Whenever symmetry does not hold, we will say that the environment is \emph{asymmetric}.

\section{Independent Information Structures}

A first step in the analysis of optimal information design is a focus on
independent signals. With independent signals, the signal $s_{i}$ of each
bidder $i$ contains all the information about the CTR $r_{i}$ that the
auctioneer releases before the auction. Thus, bidder $i$ could not learn
anything more about the true CTR from any other bidder. In turn, the
information that bidder $i$ receives from the auctioneer is the maximal
information that is available before the auction to place an informed bid.

\subsection{Independent Signals and Two Bidders\label{ist}}

We begin the analysis with two bidders and then generalize the insight to
many bidders.



\begin{proposition}[Independent and Calibrated Signalling]
\label{ics} 
In a two-bidder environment, the expected revenue of
an independent and calibrated information structure cannot exceed the one from no-disclosure.
\end{proposition}

\begin{proof}
We start by computing the expected revenue given signals $s_1$ and $s_2$. If $v_1 s_1 \geq v_2 s_2$, then the revenue can be written as:
$$ \E \left[ r_1 \cdot \frac{v_2 s_2}{s_1} \middle\vert s_1, s_2 \right] =
\E \left[ r_1  \middle\vert s_1, s_2 \right] \cdot \frac{v_2 s_2}{s_1} =
\E \left[ r_1  \middle\vert s_1 \right] \cdot \frac{v_2 s_2}{s_1} =
s_1 \cdot \frac{v_2 s_2}{s_1}  = v_2 s_2,$$
where the second equality follows from the independence of the signaling scheme and the third equality follows from calibration.

Therefore, we can write the expected revenue as:
$$\Rev = \E[\min(v_1 s_1, v_2 s_2)] \leq \min(\E[v_1 s_1], \E[v_2 s_2]) =
\min(v_1 \E[s_1], v_2 \E[s_2]) = \min(v_1 \E[r_1], v_2 \E[r_2]),$$
where the first inequality follows from Jensen's inequality and the concavity of the minimum. The last equality follows from calibration. Finally, note that $\min(v_1 \E[r_1], v_2 \E[r_2])$ is the revenue from no-disclosure.
\end{proof}

We have thus shown that full-disclosure can never revenue-dominate no-disclosure. We next show that generically
full-disclosure is, in fact, strictly revenue-dominated by zero no-disclosure.



\begin{corollary}
[Zero vs.~Complete Information Disclosure]
\label{zc} 
In the two-bidder environment, if both $v_1 r_1 > v_2 r_2$ and $v_1 r_1 < v_2 r_2$ occur with positive probability, then the revenue from no-disclosure strictly dominates the revenue from full-disclosure.
\end{corollary}

\begin{proof}
If both events $v_1 r_1 > v_2 r_2$ and $v_1 r_1 < v_2 r_2$ occur with positive probability, then Jensen's inequality holds with a strict inequality $\E[\min(v_1 r_1, v_2 r_2)] < \min(\E[v_1 r_1], \E[v_2 r_2])$ in the previous proof.
\end{proof}

The argument suggests that the revenue dominance result does not extend to more than two bidders.
With more than two bidders, the smaller of the two highest realizations
determines the price, and the expectation of the smaller of the two highest is
now larger than the unconditional expectation of the second highest
click-through rate.

\subsection{Independent Signals with Many Bidders\label{sec:3indep}}

Indeed, the power of independent signalling is much improved in the presence
of competition, and we now consider the case of more than two bidders, $n>2$.
We show by example that an independent symmetric-calibration signal can
improve the revenue, and thus partial revelation is better than no- or
full-disclosure.

Consider a symmetric three-bidder environment with $v_1 = v_2 = v_3 = 1$. For each bidder $i$ let $r_i = 0$ with probability $2/3$, and $r_i = 1$ with probability $1/3$. Moreover, assume that $r_1$, $r_2$ and $r_3$ are independent. It is simple to check that full revelation has 
revenue $7/27$ and no-revelation has 
revenue $9/27$. This can be improved by the following signaling scheme with partial bundling:

\begin{center}
\begin{tabular}
[c]{l|c|c}%
\toprule
\backslashbox{$r_i$}{$s_i$} & $s_i = 0$ & $s_i = 4/9$ \\\hline
$r_i = 0$ & $1/4$ & $5/12$ \\
$r_i = 1$ & $0$ & $1/3$\\
\bottomrule
\end{tabular}
\end{center}

The revenue of partial disclosure is $4/9$ whenever at least two of the bidders have the high signal, which happens with probability $1-(1/4)^3 - 3 (1/4)^2 (3/4) = 27/32$. Hence, the overall revenue of this partial disclose scheme is $3/8$, which dominates both full-disclosure and no-disclosure.\\

Instead of trying to optimize for the optimal independent information structure, we will move to the more powerful model of correlated information structures in the next section.



The results here mirror earlier results by \citet{boar09}, who considers an
ascending auction in a private value setting without click-through
rates. In his analysis, he restricts attention to independent signals and
establishes that with two bidders, the seller's revenue is smaller \emph{with%
} than \emph{without} information disclosure. He further shows that in a symmetric model, as the number of bidders become arbitrarily large, complete
information disclosure eventually revenue-dominates zero information
disclosure.

\section{Correlated Structures in Symmetric Environments\label{corr}}

We obtain a much stronger result if we allow the information structure to be
correlated across bidders.
The information flow allows influence over the
level of competition to some extent. This allows us to conflate auction items
and restore some market thickness (see \citet{lemi10}).

\subsection{A First Example: Flipping the Square}\label{sec:square}

To showcase the power of correlated information structures, we start with an example where the optimal structure is rather counterintuitive. Consider two bidders  with values $v_1 = v_2 = 1$ and independent click-through rates distributed uniformly in $\{1/2,1\}$. Or rather: the vector $r$ is uniformly distributed in $\{(1/2,1/2), (1/2,1), (1,1/2), (1,1)\}$. With an independent signaling scheme, the optimal information structure is no-disclosure, which yields revenue equal to $3/4$. With a correlated signaling scheme, one can obtain arbitrarily close to $7/8$ revenue, which is optimal since it corresponds to the welfare $\E[\max_i v_i r_i]$ of the optimal allocation.

The information structure is described in the following table:

\medskip

\begin{center}%
\begin{tabular}
[c]{l|c|c|c|c}%
\toprule
\backslashbox{$r$}{$s$} & $(3/4-\epsilon,3/4-\epsilon)$ &
$(3/4-\epsilon,3/4+\epsilon)$ & $(3/4+\epsilon,3/4-\epsilon)$ &
$(3/4+\epsilon,3/4+\epsilon)$\\\hline
$\left(  1/2,1/2\right)  $ & $\epsilon$ & $0$ & $0$ & $1/4-\epsilon
$\\
$(1/2,1)$ & $0$ & $1/4$ & $0$ & $0$\\
$(1,1/2)$ & $0$ & $0$ & $1/4$ & $0$\\
$(1,1)$ & $1/4-\epsilon$ & $0$ & $0$ & $\epsilon$\\
\bottomrule
\end{tabular}
\end{center}

\medskip

To see that this information structure is calibrated, observe that $s_1 = 3/4-\epsilon$ with probability $1/2$. This probability event can be decomposed in two: with probability $1/4 + \epsilon$ we output this signal with $r_1 = 1/2$, and with the remaining  $1/4 - \epsilon$ probability we have $r_1 = 1$. Hence:
$$\E[r_1 | s_1 = 3/4-\epsilon] = \frac{(1/4 + \epsilon) \cdot 1/2 + (1/4 - \epsilon)\cdot 1}{1/2} = 3/4-\epsilon.$$

The counterintuitive nature of this mapping can best be seen when depicted as in Figure \ref{fig:flipping_the_square}. We map the CTRs in $\{1/2,1\}$ to two values $\{3/4 - \epsilon, 3/4 + \epsilon \}$ around the mean. The symmetric pairs,
$\left(  1/2,1/2\right)  $ and $\left(  1,1\right)  $, are mapped with high
probability into symmetric, but order-reversed pairs, $(3/4+\epsilon
,3/4+\epsilon)$ and $(3/4-\epsilon,3/4-\epsilon)$, respectively. The
calibration is nonetheless achieved by the off-diagonal pairs $\left(
1/2,1\right)  $ and $\left(  1,1/2\right)  $ that are mapped into order-preserving signals with probability 1, $(3/4-\epsilon,3/4+\epsilon)$
and $(3/4+\epsilon,3/4-\epsilon),$ respectively. The $\epsilon$
perturbation in the mapping of the diagonal pairs then achieves the ordering
of the signals.

\begin{figure}[h]
\centering
\begin{tikzpicture}[scale=7]

\draw (0,0) -- (1,0) -- (1,1) -- (0,1) -- cycle;

\draw (.45, .45) -- (.55,.45) -- (.55,.55) -- (.45,.55) -- cycle;

\node[circle,fill,inner sep=1pt] at (0,0) {};
\node[circle,fill,inner sep=1pt] at (0,1) {};
\node[circle,fill,inner sep=1pt] at (1,0) {};
\node[circle,fill,inner sep=1pt] at (1,1) {};
\node[circle,fill,inner sep=1pt] at (.45, .45) {};
\node[circle,fill,inner sep=1pt] at (.55, .45) {};
\node[circle,fill,inner sep=1pt] at (.45, .55) {};
\node[circle,fill,inner sep=1pt] at (.55, .55) {};

\draw[->, line width=1pt, color=blue] (0,1) -- (.45,.55);
\draw[->, line width=1pt, color=blue] (1,0) -- (.55,.45);
\draw[->, line width=1pt, color=blue] (0,0) to [bend right] (.55,.55);
\draw[->, line width=1pt, color=blue] (1,1) to [bend right] (.45,.45);
\end{tikzpicture}
\caption{Depiction of the ``flipping the square'' structure (with $\epsilon$-flows omitted).}
\label{fig:flipping_the_square}
\end{figure}
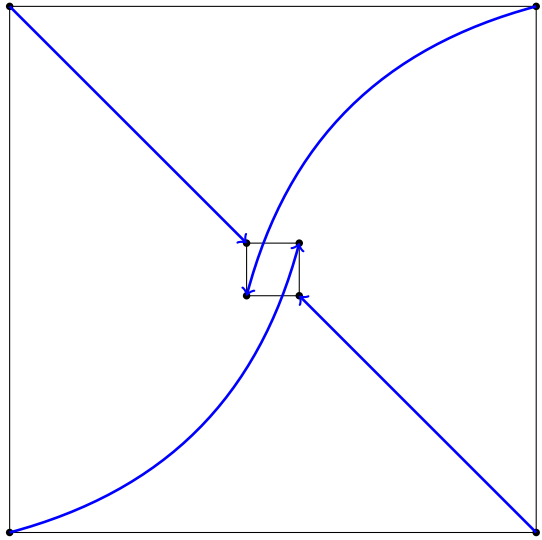

We notice the click-through signals $s_{i}$ maintain the efficient ranking of
the alternatives, and thus guarantee an efficient outcome in the auction. The
revenue in the auction is given by:%
\[
\frac{1}{4} \cdot \frac{1}{2}+\frac{1}{4} \cdot \frac{\frac{3}{4}-\epsilon}{\frac{3}%
{4}+\epsilon}+\frac{1}{4} \cdot \frac{\frac{3}{4}-\epsilon}{\frac{3}%
{4}+\epsilon}+\frac{1}{4}=\frac{21-4\epsilon}{24+32\epsilon}%
{= \frac{7}{8} - O(\epsilon)}\text{,}%
\]
which means that almost the entire surplus is extracted. The auction uses a
\emph{uniform} tie-breaking rule, thus allocating the object with equal
probability if the signals are equal across the bidders.

\subsection{A Second Example: Dispersion Along the Diagonal}

The second example maintains symmetry across the bidders but has correlated click-through rates. The resulting information structure is more subtle, reflecting the
need to balance information necessary to support an efficient allocation with information to support competition.


We present the construction, which we refer to as ``dispersion along the diagonal,'' for a small number of signals. Our main result in this section (Theorem~\ref{css}) builds on a generalization of this construction to more signals (see Lemma~\ref{cs} and Figure~\ref{fig:css}).

Consider again two bidders with values $v_1 = v_2 = 1$ and click-through rates either $(1/2,1)$ or $(1,1/2)$, with probability $1/2$ each. The CTRs are thus perfectly negatively correlated and the
social surplus is $1$. The revenue under full-disclosure would be $1/2$,
and under no-disclosure it would be $3/4$. With no-disclosure, the price-per-click
is always competitive, as $\mathbb{E}\left[  r_1\right]  /\mathbb{E}\left[
r_2\right]  =1$, but the auction fails to lead to the efficient allocation
with probability $1/2$. With the following information flow, we attain a
revenue of $0.79 > 3/4$:

\medskip

\begin{center}
\begin{tabular}
[c]{l|c|c|c|c|c|c}%
\toprule
\backslashbox{$r$}{$s$} & $(0.6,0.6)$ & $(0.6,0.75)$ & $(0.75,0.6)$ &
$(0.75,0.9375)$ & $(0.9375,0.75)$ & $(0.9375,0.9375)$\\\hline
$(1/2,1)$ & $2/15$ & $2/5$ & $0$ & $2/5$ & $0$ & $1/15$\\
$(1,1/2)$ & $2/15$ & $0$ & $2/5$ & $0$ & $2/5$ & $1/15$\\%
\bottomrule
\end{tabular}
\end{center}

\medskip

This information flow lowers the probability of an inefficient allocation from
$1/2$ to $1/5$ and attains an equilibrium price closer to 1. In particular,
\[
\min\left\vert \frac{r_{i}}{r_{j}}\right\vert =\frac{1}{2}<\frac{4}{5}%
=\min\left\vert \frac{s_{i}}{s_{j}}\right\vert \text{.}%
\]



The information flow in 
this example generates some
symmetric click-through \emph{signals} in the absence of symmetric
click-through \emph{rates}. The symmetric signals in the presence of
asymmetric rates create some inefficiency in the allocation. But the symmetric
click-through rates create the basis for signals that are adjacent, in the
sense that they are nearby, yet signal the correct ranking of the underlying
click-through rates. If we increase the numbers of signals in the construction
of the information flow, we can then reduce the revenue loss and bring it
arbitrarily close to zero. This is the following content of Lemma \ref{css}.



\subsection{Optimal Information Structure}

We can now state and establish the first main result, showing that
for any $n$-bidder symmetric environment it is possible to construct an
information structure extracting revenue that is arbitrarily close to the
optimal surplus.



\begin{theorem}[Full Surplus Extraction in Symmetric Environments]
\label{cs}For every symmetric $n$-bidder environment, there exists a
randomized and calibrated correlated information structure whose revenue is arbitrarily close to full surplus extraction. 
\end{theorem}

As a building block, we will consider the special
case of a symmetric environment of two bidders where $v_1 = v
_2 = v$ and
\begin{equation*}
\Pr[r = (l,h)]=\Pr[r = (h,l)]=1/2, 
\end{equation*}
for two values $0\leq l < h \leq 1$. We will then reduce the general symmetric
case to a composition of information structures for pairs $(h,l)$. The optimal information structure will be to disperse the signals along the diagonal as depicted in Figure \ref{fig:css}.



\begin{lemma}
[Dispersion Along the Diagonal]
\label{css} Consider the symmetric setting with two bidders with values $v_1 = v_2 = v$ where the click-through rate vector is either $(l,h)$ or $(h,l)$, each with probability $1/2$. Then, for every $\epsilon >0$, there is a calibrated, correlated information structure with revenue $vh - \epsilon$.
\end{lemma}
\begin{figure}[h]
\centering
\begin{tikzpicture}[scale=7]

\draw (0,0) -- (1,0) -- (1,1) -- (0,1) -- cycle;

\draw (.35, .35) -- (.35, .45) -- (.55,.45) -- (.55,.65) -- (.65,.65) -- (.65, .55) -- (.45, .55) -- (.45, .35) -- cycle;

\node[circle,fill,inner sep=1pt] at (0,1) {};
\node[anchor=east] at (0,1) {$(l, h)$};

\node[circle,fill,inner sep=1pt] at (1,0) {};
\node[anchor=west] at (1,0) {$(h, l)$};
\node[circle,fill,inner sep=1pt] at (.55, .45) {};
\node[circle,fill,inner sep=1pt] at (.45, .55) {};

\node[circle,fill,inner sep=1pt] at (.45, .35) {};
\node[circle,fill,inner sep=1pt] at (.35, .45) {};
\node[circle,fill,inner sep=1pt] at (.35, .35) {};

\node[circle,fill,inner sep=1pt] at (.65, .55) {};
\node[circle,fill,inner sep=1pt] at (.55, .65) {};
\node[circle,fill,inner sep=1pt] at (.65, .65) {};

\draw[->, line width=1pt, color=blue] (0,1) to [bend left] (.65,.65);
\draw[->, line width=1pt, color=blue] (0,1) -- (.55,.65);
\draw[->, line width=1pt, color=blue] (0,1) -- (.45,.55);
\draw[->, line width=1pt, color=blue] (0,1) -- (.35,.45);
\draw[->, line width=1pt, color=blue] (0,1) to [bend right] (.35,.35);

\draw[->, line width=1pt, color=blue] (1,0) to [bend right] (.65,.65);
\draw[->, line width=1pt, color=blue] (1,0) -- (.65,.55);
\draw[->, line width=1pt, color=blue] (1,0) -- (.55,.45);
\draw[->, line width=1pt, color=blue] (1,0) -- (.45,.35);
\draw[->, line width=1pt, color=blue] (1,0) to [bend left] (.35,.35);
\end{tikzpicture}
\caption{Depiction of the ``dispersion along the diagonal'' structure for Lemma~\ref{css}.}
\label{fig:css}
\end{figure}
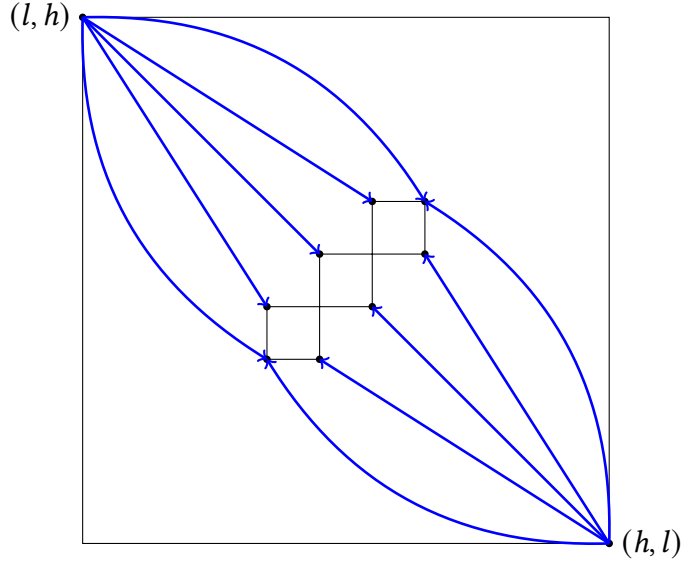
\begin{proof}
We assume without loss of generality that $v=1$. To prove the lemma, we construct an information structure with a finite set of signals, $S \subset [0, 1]$. The key to this construction is to (i) properly select the signal set $S$, and (ii) come up with a discretized and calibrated information structure $x(r, s)$ for $r \in \{(l, h), (h, l)\}$ and $s \in S$ that achieves almost optimal revenue.

We consider the following construction with parameters $\delta > 0$ and $x_0 > 0$ to be determined later (see Figure~\ref{fig:css} for an illustration).
\begin{enumerate}
  \item Signal set $S = \{s_{-K}, \ldots, s_0, s_1, \ldots, s_K\}$, where $s_0 = (l + h) / 2$, $s_k = s_0 \cdot (1 + \delta)^k$ for $-K \leq k \leq K$,
  \[K = \left\lfloor\log_{(1+\delta)}\frac{2h}{l+h}\right\rfloor - 1;\]
  \item $x((l, h), (s_k, s_{k+1})) = x((h, l), (s_{k+1}, s_k)) = x_k$ for $-K \leq k \leq K-1$, where
  \begin{align*}
    x_{k}=\frac{h-s_{k}}{s_{k}-l}\cdot x_{k-1}=x_{0}\prod_{\kappa=1}^{k}
    \frac{h-s_{\kappa}}{s_{\kappa}-l},&~\text{when}~1 \leq k \leq K-1, \\
    x_k=\frac{s_{k+1}-l}{h-s_{k+1}}\cdot x_{k+1}=x_{0}\prod_{\kappa=k}^{-1}\frac{s_{\kappa+1}-l}{h-s_{\kappa+1}},&~\text{when}~-K \leq k \leq -1;
  \end{align*}
  \item $x((l, h), (s_{-K}, s_{-K})) = x((h, l), (s_{-K}, s_{-K})) = y$, where
  \[y = \frac{s_{-K}-l}{l+h-2s_{-K}}\cdot x_{-K};\]
  \item $x((l, h), (s_K, s_K)) = x((h, l), (s_K, s_K)) = z$, where
  \[z = \frac{h-s_{K}}{2s_{K}-l-h}\cdot x_{K-1}.\]
\end{enumerate}
In the rest of the proof, we first verify that the construction is a valid calibrated and correlated information structure, then show that by choosing a sufficiently small $\delta$, the revenue is at least $h - \epsilon$.\\

\noindent\emph{Step 1: We verify that the signals in $S$ are valid probabilities, i.e., $S \subset [0, 1]$.}\\

For sufficiently small $\frac{h - l}{3(h + l)} > \delta > 0$, $K - 1 = \big\lfloor\log_{(1+\delta)}\frac{2h}{l+h}\big\rfloor - 1 \geq 1$. For all $-K \leq k \leq K$,
\begin{align*}
  s_k &\leq s_K = \frac{l+h}{2} \cdot (1+\delta)^{K} \leq \frac{l+h}{2} \cdot \frac{2h}{l + h} \cdot \frac1{1 + \delta} = \frac{h}{1+\delta} < h \leq 1; \\
  s_k &\geq s_{-K} = \frac{l+h}{2} \cdot (1+\delta)^{-K} \geq \frac{l+h}{2} \cdot \frac{l + h}{2h} \cdot (1 + \delta) \geq \frac{4hl}{4h} \cdot (1+\delta) > l \geq 0.
\end{align*}
Therefore, $S \subset [0, 1]$ is a valid finite signal space.\\

\noindent\emph{Step 2: We verify that the parameters $x_k, y, z$ are non-negative.}\\

Since $s_k \in (l, h)$, by the construction of $x_k$ for $k \neq 0$, $x_k / x_0 > 0$. For $y$ and $z$, since $s_{-K} < s_0 = (l + h) / 2 < s_K$,
\begin{gather*}
  y / x_{-K} = \frac{s_{-K}-l}{l+h-2s_{-K}} = \frac{(s_{-K} - l) / 2}{(l + h) / 2 - s_{-K}} > 0, \quad
  z / x_{K-1} = \frac{h-s_{K}}{2s_{K}-l-h} = \frac{(h-s_{K}) / 2}{s_{K} - (l + h) / 2} > 0.
\end{gather*}
Therefore, as long as $x_0 > 0$, all probability terms are positive.

\noindent\emph{Step 3: We verify that we can choose $x_0$ such that:}\\
$$\sum_s x((l, h), s) = \sum_s x((h, l), s) = 1/2$$

As we showed, the coefficients, $x_k / x_0$, $y / x_0$, and $z / x_0$ are all positive and fixed. Then with $x_0$ defined below,
\begin{gather*}
  x_0 = \frac12 \cdot \frac1{y / x_0 + z / x_0 + \sum_{k} x_k / x_0} > 0,
\end{gather*}
we have
$
  \sum_s x((l, h), s) = \sum_s x((h, l), s) = y + z + \sum_k x_k = 1/2.
$ \\

\noindent\emph{Step 4: We check that the calibration constraints are satisfied.}\\

For $-K + 1 \leq k \leq K - 1$, we verify the calibration constraint for sending signal $s_k$ to bidder $1$ as follows:
\begin{align*}
  {}&{}\sum_{s \in S} x((l, h), (s_k, s)) \cdot (l - s_k) + x((h, l), (s_k, s)) \cdot (h - s_k)
  = x_k \cdot (l - s_k) + x_{k-1} \cdot (h - s_k)
  \\
  ={}&{} x_{k-1} \cdot \frac{h - s_k}{s_k - l} \cdot (l - s_k) + x_{k-1} \cdot (h - s_k) = 0.
\end{align*}
When sending signal $s_K$ to bidder $1$:
\begin{align*}
  {}&{}\sum_{s \in S} x((l, h), (s_K, s)) \cdot (l - s_K) + x((h, l), (s_K, s)) \cdot (h - s_K)
  = z \cdot (l - s_K) + x_{K-1} \cdot (h - s_K) + z \cdot (h - s_K)
  \\
  ={}&{} \frac{h-s_{K}}{2s_{K}-l-h}\cdot x_{K-1} \cdot (l - s_K) + x_{K-1} \cdot (h - s_K) + \frac{h-s_{K}}{2s_{K}-l-h}\cdot x_{K-1} \cdot (h - s_K) = 0.
\end{align*}
When sending signal $s_{-K}$ to bidder $1$:
\begin{align*}
  {}&{}\sum_{s \in S} x((l, h), (s_{-K}, s)) \cdot (l - s_{-K}) + x((h, l), (s_{-K}, s)) \cdot (h - s_{-K})  \\
  ={}&{} x_{-K} \cdot (l - s_{-K}) + y \cdot (h - s_{-K}) + y \cdot (h - s_{-K})
  \\
  ={}&{} x_{-K} \cdot (l - s_{-K}) + \frac{s_{-K}-l}{l+h-2s_{-K}} \cdot x_{-K} \cdot (l - s_{-K}) +  \frac{s_{-K}-l}{l+h-2s_{-K}} \cdot x_{-K} \cdot (h - s_{-K}) = 0.
\end{align*}

As the construction is symmetric for bidder $1$ and $2$, we omit the verification of the calibration constraints for bidder $2$.\\

\noindent\emph{Step 5: Now that we verified this is a calibrated signaling scheme, we bound its revenue.}\\

Note that when $s \in \{(s_k, s_{k+1}), (s_{k+1}, s_k)\}_{k=-K}^{K-1}$, the auction allocates the item efficiently and the auctioneer extracts almost all the surplus. More specifically, when $s = (s_k, s_{k+1})$ for CTR profile $(l, h)$, or $s = (s_{k+1}, s_k)$ for CTR profile $(h, l)$, the conditional expected revenue is
\[h \cdot \frac{s_k}{s_{k+1}} = h / (1 + \delta) > h - \epsilon / 2,~\text{when}~\delta < \frac{\epsilon}{2h - \epsilon}.\]

Therefore, we remain to prove that the probability of not extracting revenue $h / (1 + \delta)$ is sufficiently small, i.e., $y, z < \epsilon / 8$.

Recall that $h \geq s_K \cdot (1 + \delta)$, with sufficiently small $\delta < (h - l) / 2(h + l)$,
\begin{align*}
  z {}&{}= x_{K-1} \cdot \frac{h-s_{K}}{2s_{K}-l-h} \leq x_{K-1}\cdot\frac{h-h/(1+\delta)}{2h/(1+\delta)-l-h} \\
  &{}= x_{K-1} \cdot \frac{\delta h}{h-l-\delta \cdot (h+l)} < \delta \cdot \frac{2h}{h - l} \cdot x_{K-1} < \delta \cdot \frac{2h}{h - l},
\end{align*}
which is less than $\epsilon / 8$ when $\delta < \frac{h - l}{16h} \cdot \epsilon$.

Similarly, $s_{-K} = (1 + \delta)^{-K} \cdot (l + h) / 2 \leq \frac{l + h}{2h} \cdot (1 + \delta) \cdot (l + h) / 2 < (1 + \delta)(h + 3l)/4$, with sufficiently small $\delta < (h - l) / 2(h + 3l)$,
\begin{align*}
  y {}&{}= x_{-K} \cdot \frac{s_{-K}-l}{l+h-2s_{-K}} \leq
  x_{-K} \cdot \frac{(1 + \delta)(h + 3l)/4-l}{l+h-(1 + \delta)(h + 3l)/2}  \\
  &{}= x_{-K} \cdot \frac{h - l + \delta \cdot (h + 3l)}{2h-l- 2\delta \cdot (h+3l)} < \frac32 \cdot x_{-K}.
\end{align*}

It suffices to show $x_{-K} < \epsilon / 12$ with a sufficiently small $\delta$.
Note that $s_{-K} < \cdots < s_{-1} < s_0 = (l + h) / 2 < s_1 < \cdots < s_K$, we then have $x_0 > x_1 > \cdots > x_{K-1}$ and $x_0 > x_{-1} > \cdots > x_{-K}$. Then
\[1/2 = y + z + \sum_{k = -K}^{K-1} x_k > (K + 1) \cdot x_{-K}.\]
Therefore, when $\delta < \frac{\epsilon}{6} \cdot \log \frac{2h}{l + h}$, $x_{-K}$ can be bounded by $\epsilon / 12$:
\[x_{-K} < 1/2(K+1) < \frac1{2\log_{(1 + \delta)}(2h/(l+h))} = \frac{\log(1 + \delta)}{2\log (2h / (l + h))} < \frac{\delta}{2\log (2h / (l + h))}.\]

In summary, for any given $\epsilon > 0$, we can conclude the proof with a sufficiently small $\delta$:
\begin{gather*}
  \delta < \frac{h - l}{3(h + l)} \Imply K - 1 \geq 1, \\
  \delta < \frac{\epsilon}{2h - \epsilon} \Imply h/(1 + \delta) \geq h - \epsilon/2, \\
  \delta < \min\left(\frac{h - l}{2(h + l)}, \frac{h - l}{16h} \cdot \epsilon\right) \Imply z < \epsilon / 8, \\
  \delta < \min\left(\frac{h - l}{2(h + 3l)}, \frac{\epsilon}{6} \cdot \log \frac{2h}{l + h}\right) \Imply y < \epsilon / 8. \qedhere
\end{gather*}
\end{proof}

Lemma~\ref{css} is stated for a very special case within the class of
symmetric environments. Both bidders have value $1$, each bidder has only one of \emph{two possible}
click-through rates, and the click-through rates are \emph{perfectly negatively}
correlated. However, the result can now be extended immediately to a general
symmetric environment.
The extension is based on two simple observations:

\begin{corollary}[High-Low Pairing]\label{cor:n_bidder_symmetric}
Consider $n$ bidders with values $v_i \equiv v$ and click-through rates uniformly distributed between two profiles where $r$ and $r'$ are such for two bidders $i,j \in [n]$ that we have: $r_i = r'_j > r_j = r'_i \geq r_k, r'_k$ for any $k \neq i,j$. Then, for any $\epsilon > 0$ there is a calibrated, correlated information structure with revenue $v r_i - \epsilon$.
\end{corollary}

\begin{proof}
Treat bidders $i$ and $j$ as the high/low pair in Lemma \ref{css} and do full-disclosure for any other bidder. The signal is still calibrated and only $i$ and $j$ win the item since their signals will be above the signals of any other bidder. Hence, the revenue bound in Lemma \ref{css} still holds.
\end{proof}

The next lemma shows that information structures can be composed, in the sense that if we decompose a distribution of click-through rates and design an information structure for each of them, we can later compose them without loss in calibration.

\begin{lemma}[Signal Composition]\label{lem:signal_composition}
Let $\G'$ and $\G''$ be distributions over click-through rate profiles $r$ of $n$ bidders and let $\F'$ and $\F''$ be corresponding calibrated information structures  given by joint distributions over vector pairs $(r,s)$ such that the $r$-marginals are $\G'$ and $\G''$ respectively.

Let $\G$ be the distribution obtained by sampling from $\G'$ with probability $\lambda$, and $\G''$ with probability $1-\lambda$. Define a distribution $\F$ similarly. Then, $\F$ is a calibrated information structure for $\G$ and $$\Rev(\F) = \lambda \Rev(\F') + (1-\lambda) \Rev(\F'').$$
\end{lemma}

\begin{proof}
The $r$-marginal of $\F$ is clearly $\G$ and $\Rev(\F) = \lambda \Rev(\F') + (1-\lambda) \Rev(\F'')$. The only non-trivial part is to check that $\F$ is calibrated, which we do below:
\begin{align*}
\E_\F[r_i | s_i = s'_i] & =
\frac{\E_\F[r_i \one \{s_i - s'_i\}]}{\Pr_\F(s_i = s'_i)}
= \frac{\lambda\E_{\F'}[r_i \one \{s_i - s'_i\}] + (1-\lambda) \E_{\F''}[r_i \one \{s_i - s'_i\}]
}{\lambda\Pr_{\F'}(s_i = s'_i) + (1-\lambda) \Pr_{\F''}(s_i = s'_i)} \\
& = \frac{\lambda\E_{\F'}[s'_i \one \{s_i - s'_i\}] + (1-\lambda) \E_{\F''}[s'_i \one \{s_i - s'_i\}]
}{\lambda\Pr_{\F'}(s_i = s'_i) + (1-\lambda) \Pr_{\F''}(s_i = s'_i)} = s'_i. \qedhere
\end{align*}
\end{proof}



\begin{proof}[Proof of Theorem~\ref{cs}]
Consider an $n$-bidder symmetric environment and assume for simplicity that the distribution over CTRs is discrete. For every profile of CTRs where bidder $i$ has the highest CTR $r_i$ and bidder $j$ has the second highest CTR $r_j$ (breaking ties lexicographically), we can pair with a profile where the CTRs of $i$ and $j$ are reversed. This leads to a decomposition of the original distribution of CTRs into distributions with support two of the form studied in Corollary \ref{cor:n_bidder_symmetric}. The results follow from applying Corollary \ref{cor:n_bidder_symmetric} together with the composition technique in Lemma \ref{lem:signal_composition}.
\end{proof}

It turns out that the idea behind the construction in this section can be generalized to work under the weaker requirement of \emph{equal means}. The main building block for this is the following generalization of Lemma~\ref{css}, which we prove in Appendix~\ref{app:equal-means}.


\begin{proposition}
[Generalized Dispersion]\label{equal-means}
Consider a setting with two bidders with $v_1 = v_2 = v$ and a support size two distribution over click-through rate vectors $(r_1, r_2)$ such that $\E[r_1] = \E[r_2]$. Then for every $\epsilon$, there is a calibrated and correlated information structure with revenue $v\cdot \E[\max(r_1, r_2)] - \epsilon$.
\end{proposition}

\section{Correlated Structures Beyond Symmetric Environments}\label{sec:asymmetric}


Our analysis so far has maintained a weak notion of ex-ante symmetry among the bidders. In Proposition \ref{equal-means}, the bidders were assumed to have the same ex-ante mean in click-through rates, but may have different supports and different distributions in terms of the ex-post click-through rates. In the presence of ex-ante symmetry we established that an optimal information design supported a revenue level equal to the full social surplus. In this section, we pursue a more limited objective for general asymmetric environments. The central result of this section, Theorem \ref{thm:moderate}, establishes that optimal information design remains a powerful instrument to increase revenue, irrespective of the joint distribution of the click-through rates. In particular, we can prove that there always exist information structures that revenue-dominate either no-disclosure or full-disclosure information structures.

For this general environment, we do not provide a complete identification of the optimal information design across all possible configurations of the joint distributions of click-through rates. However, the arguments developed so far will be sufficient to offer significant qualitative characterizations of optimal information design.

Different information structures offer different levels of informativeness, and in fact form a \emph{lattice} (see Figure~\ref{fig:lattice}).
The no-disclosure information structure is the minimal information policy. The
full-disclosure information structure is the maximal information policy. Together,
they form the set of \emph{extremal} informational policies. We refer to every
information policy that is not extremal as \emph{moderate}, and to any information policy that does not consist of a combination of full- or no-disclosure as \emph{interior}.

Theorem \ref{thm:moderate} establishes that in all environments, symmetric or asymmetric, there exists a moderate information structure that presents a strict improvement. We then show that in some sense this is the strongest result that we can obtain in a general asymmetric environment. Theorem \ref{thm:uni-con-weak} establishes that in an environment in which one bidder always has the highest click-through rate, the optimal information structure is moderate, but not interior. Finally, Proposition \ref{pro:uni-inc} shows that in other environments the optimal information structure must indeed be not only moderate, but in fact interior.

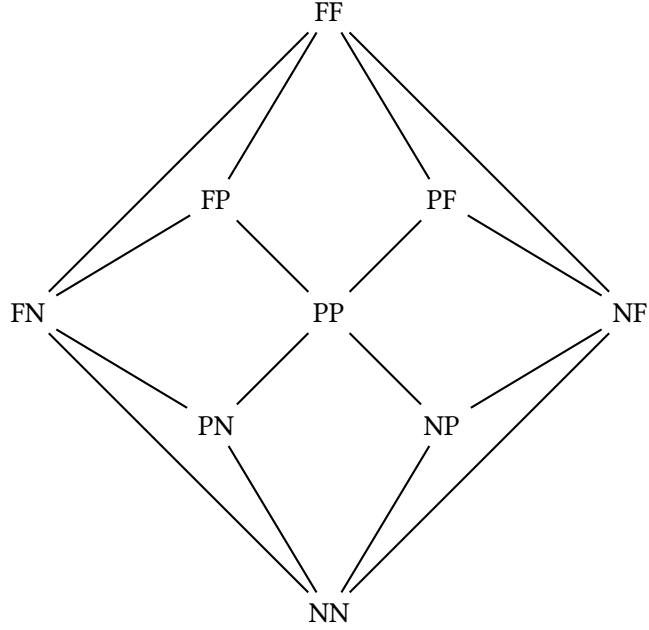
\begin{figure}[ht]
\begin{center}
\begin{tikzpicture}[scale=2]
\node (FN) at (0,2) {\small FN};
\node (PN) at (1.25,1.25) {\small PN};
\node (FP) at (1.25,2.75) {\small FP};
\node (FF) at (2,4) {\small FF};
\node (PP) at (2,2) {\small PP};
\node (NN) at (2,0) {\small NN};
\node (PF) at (2.75,2.75) {\small PF};
\node (NP) at (2.75,1.25) {\small NP};
\node (NF) at (4,2) {\small NF};
\draw[-,thick] (FN) -- (FF);
\draw[-,thick] (FN) -- (FP);
\draw[-,thick] (FN) -- (PN);
\draw[-,thick] (FN) -- (NN);
\draw[-,thick] (NF) -- (FF);
\draw[-,thick] (NF) -- (PF);
\draw[-,thick] (NF) -- (NP);
\draw[-,thick] (NF) -- (NN);
\draw[-,thick] (NN) -- (PN);
\draw[-,thick] (NN) -- (NP);
\draw[-,thick] (FF) -- (FP);
\draw[-,thick] (FF) -- (PF);
\draw[-,thick] (PP) -- (FP);
\draw[-,thick] (PP) -- (PF);
\draw[-,thick] (PP) -- (PN);
\draw[-,thick] (PP) -- (NP);
\end{tikzpicture}
\end{center}
\caption{Depiction of the disclosure lattice for two bidders. The label XY of a vertex denotes the pair of policies X and Y used by bidder 1 and bidder 2. N stands for no-disclosure, P for partial disclosure, and F for full-disclosure. All information structures inside the square are interior.}
\label{fig:lattice}
\end{figure}

\subsection{Extremal Structures are Dominated}

For the remainder of this section we consider an environment in which there are two bidders with values $v_1 = v_2 = 1$, and two possible click-through rate configurations, namely $(r_1,r_2)$ and $(r'_1,r'_2).$ Without loss of generality, we can always label the identities of the bidders and the click-through rates so that:
\[
    r_1 \geq r'_1, r_2, r'_2.\label{eq:rank}
\]
The prior probability of the pair $(r_1,r_2)$ is $p$, the other pair $(r'_1,r'_2)$ has the complementary probability $1-p$. Let $\mu_1 = p \cdot r_1 + (1-p) \cdot r'_1$ and $\mu_2 = p \cdot r_2 + (1-p) \cdot r'_2$ be the expected click-through rates of bidder 1 and bidder 2, respectively.

The main result of this section is the following theorem.

\begin{theorem}[Moderate Information Structures]\label{thm:moderate}
There always exists a moderate information structure that strictly dominates any extremal information structure.
\end{theorem}

The theorem will be established in several steps. We split our analysis to match the nature of the competition between the bidders. In the ex-ante symmetric case, the bidder with the highest click-through rate was variable across the realizations, and we thus had a variable winner. In the asymmetric setting, the bidder with the highest click-through rate may be uniform across click-through rates. Given the ranking, 
we will split our analysis into the following cases:
$$\begin{aligned}
& \text{Uniform Winner:  } r'_1 \geq r'_2 \qquad  & & \text{Variable Winner:  } r'_1 < r'_2 \\
& \text{Congruent Loser:  } r_2 \geq r'_2 \qquad  & & \text{Incongruent Loser:  } r_2 < r'_2 \\
& \text{Weak Competition:  } r_2 \leq \mu_1 \qquad  & & \text{Strong Competition:  } r_2 > \mu_1
\end{aligned}$$

\subsection{Moderate vs. Interior Structures}

We next show by construction that Theorem \ref{thm:moderate} cannot be strengthened from moderate to interior structures in the general asymmetric setting. For this, we consider the case of the uniform winner with a congruent loser under weak competition. In this setting, bidder $1$ has the higher click-through rate in all realizations of click-through rates, and is thus the uniform winner. Moreover, the ranking across realizations is the same for all bidders, and hence congruent. Finally, the expected click-through rate of the winner is larger than the maximum of the realized click-through rates of the loser.

For this setting, we can identify the uniquely optimal information structure. Namely, it is
optimal to bundle the click-through rates of the winner and to unbundle the click-through rates of the
loser. The bundling of the click-through rates of the winner generates more competitive
prices, and hence higher revenue for the auctioneer. Thus, the optimal information structure is moderate, but not interior.

\begin{theorem}[Uniform Winner, Congruent Loser, Weak Competition]\label{thm:uni-con-weak}
The optimal information structure for the uniform winner, congruent loser, and weak competition case leaves
the loser unbundled and fully bundles the winner.
\end{theorem}

The main tool in our proof of Theorem~\ref{thm:uni-con-weak} is the following lemma. 

\begin{lemma}[Chebyshev's sum inequality   \citep{Hardy88}]\label{lem:chebyshev}
Given two sequences $a_1 \geq a_2 \geq \hdots a_n \geq 0$ and $b_1 \geq b_2 \geq \hdots b_n  \geq 0$ that are monotone in the same direction, and a set of non-negative weights $w_i \geq 0$ (not necessarily monotone), then:
$$\left( \sum_i w_i a_i b_i \right) \cdot \left( \sum_i w_i \right) \geq
  \left( \sum_i w_i a_i \right) \cdot \left( \sum_i w_i b_i \right).$$
If $\{a_i\}$ and $\{b_i\}$ sequences are monotone in different directions (one increasing and one decreasing), the inequality holds in the opposite direction.
\end{lemma}
(Aside: The probabilistic  interpretation of Lemma~\ref{lem:chebyshev} is that if $A$ and $B$ are two positively-correlated random variables, then $\E[AB] \geq \E[A] \cdot \E[B]$.)

\begin{proof}[Proof of Theorem~\ref{thm:uni-con-weak}]

We will start with a generic solution $x(r,s)$ and show that using two applications of Chebyshev's sum inequality (Lemma~\ref{lem:chebyshev}) we can bound it with respect to 
\begin{align}
p \cdot r_1 \cdot \frac{r_2}{\mu_1} + (1-p) \cdot r'_1 \cdot \frac{r'_2}{\mu_1}, \label{eq:target}
\end{align}
which is the revenue obtained from leaving the loser unbundled and bundling the winner.
We will proceed in three steps.\\

\noindent\emph{Step 1: Bounding the revenue.}
Consider any information structure defined by $x(r,s)$. For each signal $s = (s_1, s_2)$, the revenue in the event that bidder $1$ wins is $(x(r,s) r_1 + x(r',s) r'_1) \cdot s_2 / s_1$. If bidder $2$ wins, the revenue is $(x(r,s) r_2 + x(r',s) r'_2) \cdot s_1 / s_2$. In the second case, observe that since $s_2 \geq s_1$ we have $\frac{s_1}{s_2} < 1 < \frac{s_2}{s_1}$, and since we are in the uniform winner case, we know that $x(r,s) r_1 + x(r',s)r'_1 \geq x(r,s) r_2 + x(r',s)r'_2$. Hence, we can bound:
$$(x(r,s) r_2 + x(r',s) r'_2) \cdot \frac{ s_1 }{ s_2 } \leq (x(r,s) r_1 + x(r',s) r'_1) \cdot \frac{ s_2 }{ s_1 },$$ and write:
$$\Rev \leq \sum_s (x(r,s) r_1 + x(r',s) r'_1) \cdot \frac{ s_2 }{ s_1 }. $$


\noindent\emph{Step 2: Unbundling the loser.} Substituting $s_2$ by the calibration constraint we obtain:
$$\begin{aligned}
\Rev &  \leq
  \sum_{s_2} \sum_{s_1} \left( x(r,(s_1,s_2))\frac{r_1}{s_1} + x(r',(s_1,s_2))\frac{r'_1}{s_1} \right) \cdot
  \frac{\sum_{s_1} x(r,(s_1,s_2))r_2 + x(r',(s_1,s_2))r'_2}{\sum_{s_1} x(r,(s_1,s_2) +  x(r',(s_1,s_2))} \\ &
\leq
  \sum_{s_2} \sum_{s_1} \left(x(r,(s_1,s_2))\frac{r_1r_2}{s_1} + x(r',(s_1,s_2))\frac{r'_1r'_2}{s_1}\right),
  \end{aligned}$$
where the second inequality follows from Chebyshev's sum inequality with $$\{a_i\}_i = \left( \frac{r_1}{s_1^1}, \hdots, \frac{r_1}{s_1^n}, \frac{r'_1}{s_1^1}, \hdots, \frac{r'_1}{s_1^n}  \right) \quad \text{and} \quad \{b_i\}_i = (r_2, \hdots,r_2, r'_2, \hdots, r'_2).$$ Congruence implies that the sequence $\{b_i\}$ is sorted. The sequence $\{a_i\}$ is sorted because $r'_1 \leq s_1^1 \leq \hdots \leq s_1 \leq r_1$, and hence $\frac{r_1}{s_1} \geq 1 \geq \frac{r'_1}{s_1}$. Now that the loser is unbundled, there is no longer any need to keep track of $s_2$. To simplify notation we will define:
$$\x(r,s_1) = \sum_{s_2} x(r,(s_1,s_2))$$
and re-write our current bound on the objective as:
 $$\Rev \leq \sum_{s_1} \frac{\x(r,s_1) r_1r_2 +  \x(r',s_1)r'_1r'_2}{s_1}.$$

\noindent\emph{Step 3: Bundling the winner.} We will replace $s_1$ according to the calibration constraint in the expression above and replace:
$$\lambda(s_1) = \frac{\x(r,s_1)}{\x(r,s_1) + \x(r',s_1)} \quad \text{ and } \quad \lambda'(s_1) = \frac{ \x(r',s_1) }{ \x(r,s_1) + \x(r',s_1) }.$$
We obtain:
$$\begin{aligned}
\Rev & \leq \sum_{j} (\x(r,s_1)r_1r_2 +  \x(r',s_1)r'_1r'_2) \cdot \frac{\x(r,s_1) + \x(r',s_1)}{r_1\x(r,s_1) + r'_1\x(r',s_1)} \\ & =  \sum_{j} \frac{\lambda(s_1)r_1r_2 +  \lambda'(s_1)r'_1r'_2}{\lambda(s_1)r_1 + \lambda'(s_1)r'_1} \cdot (\x(r,s_1) + \x(r',s_1)).
\end{aligned}$$

Now we can apply Chebyshev's sum inequality one more time with:
$$a(s_1) = \frac{\lambda(s_1)r_1r_2 +  \lambda'(s_1)r'_1r'_2}{\lambda(s_1)r_1 + \lambda'(s_1)r'_1}, \quad b(s_1) = \lambda(s_1)r_1 + \lambda'(s_1)r'_1, \quad w(s_1) = \x(r,s_1) + \x(r', s_1).$$
Note that we can reorder signals $s_1$ in any order we wish. Let us reorder the signals so that $\lambda(s_1)$ is increasing.
This immediately implies that $b(s_1)$ is increasing. Namely, $b(s_1) = \lambda(s_1) r_1 + (1-\lambda(s_1)) r'_1 = \lambda(s_1) (r_1 - r'_1) + r'_1$. For $a(s_1)$, we can take the derivative in $\lambda$,
\[
\frac{d}{d\lambda} a(s_1) = \frac{r_1\cdot r'_1 \cdot (r_2-r'_2)}{(r_1\lambda - r'_1 (1-\lambda))^2} > 0.
\]
Thus, $a(s_1)$ is also increasing in $\lambda$,
allowing us to apply the inequality to obtain:
$$\sum_{s_1} w(s_1)  a(s_1) \leq \left(\sum_{s_1} w(s_1) a(s_1) b(s_1) \right) \frac{\sum_{s_1} w(s_1)}{ \sum_{s_1} w(s_1) b(s_1)},$$ which translates to (since $\sum_{s_1} w(s_1) = 1$):
$$\Rev \leq \frac{\sum_{s_1} \x(r,s_1)r_1r_2 +  \x(r',s_1)r'_1r'_2}{\sum_{s_1}  \x(r,s_1)r_1 +  \x(r',s_1)r'_1} = \sum_{s_1} \x(r,s_1)r_1 \frac{r_2}{\mu_1} +  \x(r',s_1)r'_1 \frac{r'_2}{\mu_1} = p r_1 \frac{r_2}{\mu_1} + (1-p) r'_1 \frac{r'_2}{\mu_1}.$$
This is the revenue obtained by bundling the winner and unbundling the loser as desired.
\end{proof}

The impact of competitive prices via competitive click-through rates can be so strong that
even a weaker position of the winning bidder can actually increase the revenue of
the auctioneer. In fact, as long as we maintain a uniform winner and a
congruent loser, it is revenue-increasing for the winner to have a lower
click-through rate in the low click-through rate.

For the next result we stay with the uniform winner setting, but now flip the ranking of click-through rates across click-through realizations. Thus, we consider the case of the incongruent loser. We can verify that the competition is guaranteed to be weak in the above sense.

Now, we can show that an interior information structure will always be the optimal information structure. In contrast, the moderate information structure that was optimal in the congruent setting can be shown to perform worse than either of the extremal information structures.

\begin{proposition}[Uniform Winner and Incongruent Loser]
\label{uwil}\label{pro:uni-inc}
With a uniform winner and incongruent loser, there is always an interior
information structure that is revenue-improving over all exterior
information structures. In particular, partially bundling the winner \emph{%
and} the loser is revenue-improving, relative to all exterior information
structures:

\begin{center}
\begin{tabular}{c|c|ccc}
\toprule
\backslashbox{$r$}{$s$} & $(r_1,x)$ & $(z,y)$ & $\left(z,x\right) $ & $%
\left( r_{1}^{\prime },y\right) $ \\ \hline
$(r_{1},r_{2})$ & $p-q$ & $q$ & $0$ & $0$ \\ \hline
$(r_{1}^{\prime },r_{2}^{\prime })$ & $0$ & $0$ & $q^{\prime }$ & $p^{\prime
}-q^{\prime }$\\
\bottomrule%
\end{tabular}%
\end{center}
for suitably chosen $x < y \leq z$.
\end{proposition}

These two cases demonstrate how rich the optimal structures in the asymmetric setting can be, but they also emphasize that the same guiding principles that governed the optimal policy in symmetric settings are in play. Namely, as in the symmetric setting, the optimal schemes that we identify are not extremal, maintain efficiency, and seek to strengthen competition through calibrated signals.

\section{Conclusion}

A large share of online advertising, whether in search advertising, display advertising,
or social networks, is allocated by various auction mechanisms. The auction
seeks to form a match between the viewer and competing advertisers. The
viewers are typically heterogeneous in many dimensions: their
characteristics, their preferences, their (past) shopping behavior, their
browsing history, and many other aspects, observable and unobservable. A
critical dimension of the heterogeneity is the click-through rate of the
user for a specific advertiser.

The disclosure policy of the auctioneer regarding the click-through rates
therefore influences the distribution of bids, holding fixed the
distribution of preferences among the bidders. By disclosing less, the
auctioneer in effect bundles certain features, or as \citet{lemi10} suggest,
conflates features of the viewer. The process of conflation influences the
thickness or thinness of the market---in other words, the strength of the
competition.

In our analysis, conflation was achieved by bundling the information
regarding click-through rates in an optimal manner. The optimal information
structure conveys just enough information so that the resulting bidding
process ranks the bidders according to the true social value of each
bidder. At the same time, the information is released only partially to
maintain bids as close as possible to a perfectly competitive level. We
show that this requires that the information is provided to each bidder as
private information at an individual level, rather than as public information
on a market level.

An alternative approach might have been to directly apply a version of the
revelation principle to the information design problem (see \citet{bemo19}).
The revelation principle for information design would suggest that the
auctioneer only requires as many signals as allocations to be implemented.
This would suggest that a public signal that informs the market about the
identity of the winner and loser would be sufficient. In a two-bidder
environment, this would suggest that a binary public information structure
would be sufficient. Indeed, the binary signal would lead to an efficient
allocation. But as our analysis and Example 2 have shown in detail, this
would lead to a lower revenue for the seller than what is attainable by
optimal information design. The apparent failure of the revelation principle
here is related to the click-through auction. We deliberately fixed the
auction mechanism to trace the implications of information design in a
given auction mechanism. But by adopting the rules of the click-through
auction, we effectively limit the message space of each type, which is
the fundamental source of failure of the revelation principle, as
famously observed by \citet{grla86}.

A significant advance in our information design problem is to allow for
multi-dimensional and private information in a strategic setting. By
contrast, the most recent result in the design of optimal information
structure requires one-dimensional, or equivalently symmetric, solutions to
optimal design (see \citet{klms20}). A general approach to optimal
multi-dimensional information design in strategic settings is a wide-open
question. In the current context, we could make progress by insights
specific to the auction setting.

Throughout this work we maintained the assumption of complete information, regarding
the value of each bidder. While this setting is commonly adopted in the
analysis of sponsored search auctions, it would clearly be an important issue
for how the optimal information policy is influenced by private information,
regarding the valuation of each bidder. We would then be in a setting where
both the auctioneer and bidders have private information. This also remains
a wide-open issue, even in a setting with a single receiver, and progress is
currently being made only in specific settings, either binary actions and states
or multiplicative separable settings (see \citet{klmz15} and \citet{cast21},
respectively). We would hope that the current auction setting would give us
additional tools to address these issues eventually.

\bibliographystyle{ACM-Reference-Format}
\bibliography{general}

\newpage
\appendix

\section{Full Surplus Extraction with Equal Means}\label{app:equal-means}


\begin{lemma*}[Lemma~\ref{equal-means} Restate]
Consider a setting with two bidders with click-through rates $(l,a)$ and $(h,b)$, where $0 \leq l < h$ and $a \neq b \geq 0$, with probability
\[
\Pr[(l,a)]=\frac{h-b}{h+a-l-b},\qquad\Pr[(h,b)]=\frac{a-l}{h+a-l-b},
\]
and the same mean
\[
\mu=\frac{ha-lb}{h+a-l-b}.
\]
Then, for every $\epsilon > 0$,  there is a calibrated, correlated information structure that extracts at least $1-O(\sqrt{\epsilon})$ fraction of the welfare as revenue.
\end{lemma*}

Note that we can, without loss of generality, assume that $a\neq l$ and $b\neq h$. Otherwise, because of the
symmetric mean, when $a=l$ we must have $b=h$, and vice versa. In this special case, the
problem is trivial (full-disclosure is optimal).


Consider the following information structure:

\begin{itemize}
\item Mapping $(l,a)$ to $(s_{2k},s_{2k+1})$ with probability $x_{k}$ for $-K
\leq k \leq K$;

\item Mapping $(h,b)$ to $(s_{2k},s_{2k-1})$ with probability $x_{k}^{\prime}$
for $-K \leq k \leq K$.
\end{itemize}

In particular,
\[
s_{i} = \left\{
\begin{array}
[c]{ll}%
h^{*} = \min\{h, \max\{a, b\}\}, & i = 2K+1;\\
s \cdot(1+\epsilon)^{i}, & -2K \leq i \leq2K;\\
l^{*} = \max\{l, \min\{a, b\}\}, & i = -2K-1.
\end{array}
\right.
\]
In general, we require that $s_{-2K-1} < s_{-2K} < s_{2K} < s_{2K+1}$, that
is
\[
l^{*} < s \cdot(1 + \epsilon)^{-2K} < s \cdot(1 + \epsilon)^{2K} < h^{*}.
\]
Note that either $a > s_{2K}, \ldots, s_{-2K} > b$ or $a < s_{2K}, \ldots,
s_{-2K} < b$.

We claim that as $\epsilon\rightarrow0$, there exists a valid information structure of the above form satisfying the calibration constraint and extracting
at least $1 - O(\sqrt{\epsilon})$ of the welfare as revenue.\\

Consider the calibration constraints: For $-K\leq k\leq K$, the first
buyer receives signal $s_{2k}$ when the CTR profiles are $(l,a)$ or $(h,b)$,
with probability $x_{k}$ and $x_{k}^{\prime}$, respectively, i.e.,
\begin{gather*}
h\cdot x_{k}^{\prime}+l\cdot x_{k}=s_{2k}\cdot(x_{k}^{\prime}+x_{k}
)\iff(h-s_{2k})x_{k}^{\prime}=(s_{2k}-l)x_{k}.
\end{gather*}
Similarly, for the second buyer receiving signal $s_{2k+1}$,
\[
a\cdot x_{k}+b\cdot x_{k+1}^{\prime}=s_{2k+1}\cdot(x_{k}+x_{k+1}^{\prime}%
)\iff(a-s_{2k+1})x_{k}=(s_{2k+1}-b )x_{k+1}^{\prime}.
\]
Therefore, we have
\[
x_{k+1} = \frac{h - s_{2k+2}}{s_{sk+2}-l}x^{\prime}_{k+1} = \frac{h -
s_{2k+2}}{s_{2k+2}-l} \cdot\frac{a-s_{2k+1}}{s_{2k+1}-b} x_{k},
\]
and similarly,
\[
x^{\prime}_{k+1} = \frac{a- s_{2k+1}}{s_{2k+1} - b}x_{k} = \frac{a- s_{2k+1}%
}{s_{2k+1} - b} \cdot\frac{h - s_{2k}}{s_{2k} - l} x^{\prime}_{k}.
\]
Note that the coefficients here are all positive,
\[
\frac{a- s_{2k+1}}{s_{2k+1} - b} > 0,
\]
because either $a > s_{2k+1} > b$ or $a < s_{2k+1} < b$.

Thus,
\begin{align*}
x_{k} = x_{0} \prod_{i=1}^{k} \frac{h - s_{2i}}{s_{2i}-l} \cdot\frac
{a-s_{2i-1}}{s_{2i-1}-b}  &  \qquad x_{-k} = x_{0} \prod_{i=1}^{k}
\frac{s_{2-2i}-l}{h - s_{2-2i}} \cdot\frac{s_{1-2i}-b}{a-s_{1-2i}}\\
x^{\prime}_{k} = x^{\prime}_{0} \prod_{i=1}^{k} \frac{h - s_{2i-2}}%
{s_{2i-2}-l} \cdot\frac{a-s_{2i-1}}{s_{2i-1}-b}  &  \qquad x^{\prime}_{-k} =
x^{\prime}_{0} \prod_{i=1}^{k} \frac{s_{-2i}-l}{h - s_{-2i}} \cdot
\frac{s_{1-2i}-b}{a-s_{1-2i}}%
\end{align*}

Therefore, the construction is valid, if and only if the following probability
constraint is satisfied:
\[
\sum x_{k} = \Pr[(l, a)] = \frac{h - b}{h + a - l - b}, \qquad\sum x^{\prime
}_{k} = \Pr[(h, b)] = \frac{a - l}{h + a - l - b},
\]
and
\[
x_{k}, x^{\prime}_{k} \geq0.
\]

It remains to argue the existence of such an information structure $(s, \epsilon,
K)$ with $\epsilon\rightarrow0^{+}$, and for the purpose of the revenue bound $x_{K}, x^{\prime}_{-K} = o(1)$. This is established in Lemma~\ref{lem:prob} and Lemma~\ref{lem:efficient} below.

\begin{lemma}
\label{lem:prob} For sufficiently small $0 < \epsilon< (\min\{h^{*} / \mu,
\mu/ l^{*}\})^{1/4} - 1$ and $K < \frac14\log_{1 + \epsilon} \min\{h^{*} /
\mu, \mu/ l^{*}\}$, there exist $s \in(s^{-}, s^{+})$ such that the above
construction is a valid signaling scheme, where
\[
s^{-} = \mu\cdot(1+\epsilon)^{-2K}, \qquad s^{+} = \mu\cdot(1+\epsilon)^{2K}.
\]

\end{lemma}

\begin{lemma}
\label{lem:efficient} As $\epsilon\rightarrow0^{+}$ and $K = \frac
{\sqrt{\epsilon}}4 \log_{1+\epsilon} \min\{h^{*} / \mu, \mu/ l^{*}\} =
O(1/\sqrt{\epsilon})$, the valid signaling scheme defined in
Lemma~\ref{lem:prob} has $x_{K}, x^{\prime}_{-K} = O(\sqrt{\epsilon})$.
\end{lemma}


\begin{proof}
[Proof of Lemma~\ref{lem:prob}]Note that from the previous analysis, we have
\[
x_{k} = A_{k}(s, \epsilon) \cdot x_{0}, \quad x^{\prime}_{k} = A^{\prime}%
_{k}(s, \epsilon) \cdot x^{\prime}_{0},
\]
where $A_{k}(s, \epsilon)$ and $A^{\prime}_{k}(s, \epsilon)$ are coefficients
depending on $s$ and $\epsilon$.

Therefore, the construction is valid if and only if the following two
conditions are satisfied:

\begin{itemize}
\item The probability constraints are satisfied, i.e.,
\begin{gather}
\label{eq:prob}x_{0} \sum_{k=-K}^{K} A_{k}(s, \epsilon) = \Pr[(l, a)], \qquad
x^{\prime}_{0} \sum_{k=-K}^{K} A^{\prime}_{k}(s, \epsilon) = \Pr[(h, b)];
\end{gather}

\item All probabilities are non-negative, or equivalently,
\begin{gather}
\label{eq:pcoeff}\forall k \in[-K, K],~A_{k}(s, \epsilon), A^{\prime}_{k}(s,
\epsilon) \geq0.
\end{gather}

\end{itemize}

For any sufficiently small $0 < \epsilon< (\min\{h^{*} / \mu, \mu/
l^{*}\})^{1/4} - 1$, there exists $1 \leq K < \frac14\log_{1+\epsilon}
\min\{h^{*} / \mu, \mu/ l^{*}\}$. Then, for any $s \in(s^{-}, s^{+})$ we have
\[
l^{*} < \mu\cdot(1 + \epsilon)^{-4K} = s^{-} \cdot(1+\epsilon)^{-2K} <
s_{-2K}, \ldots, s_{2K} < s^{+} \cdot(1+\epsilon)^{2K} = \mu\cdot(1 +
\epsilon)^{4K} < h^{*}.
\]

Hence, for all $-K \leq k \leq K$, we know that \eqref{eq:pcoeff} is satisfied.
To conclude the proof, we next argue that there exists $s \in(s^{-}, s^{+})$
such that \eqref{eq:prob} is also satisfied.

Note that $x_{k}$ and $x^{\prime}_{k}$ are inter-dependent, i.e.,
\[
x^{\prime}_{k} = x_{k} \cdot\frac{s_{2k} - l}{h - s_{2k}}.
\]

In particular, if $s_{2k} < \mu$, then
\[
\frac{s_{2k} - l}{h - s_{2k}} < \frac{\mu- l}{h - \mu} = \frac{\frac{ha- lb}{h
+ a - l - b} - l}{h - \frac{ha- lb}{h + a - l - b}} = \frac{(h - l)(a - l)}{(h
- l)(h - b)} = \frac{a - l}{h - b} = \frac{\Pr[(h, b)]}{\Pr[(l, a)]},
\]
and vice versa.%

Now,
\[
\sum_{k = -K}^{K} x^{\prime}_{k} = \sum_{k = -K}^{K} x_{k} \cdot\frac{s_{2k} -
l}{h - s_{2k}} = x_{0} \sum_{k = -K}^{K} A_{k}(s, \epsilon) \cdot\frac{s_{2k}
- l}{h - s_{2k}},
\]
and so \eqref{eq:prob} can be satisfied if and only if
\[
\frac{\sum_{k = -K}^{K} A_{k}(s, \epsilon) \cdot\frac{s_{2k} - l}{h - s_{2k}}%
}{\sum_{k = -K}^{K} A_{k}(s, \epsilon)} = \frac{\Pr[(h, b)]}{\Pr[(l, a)]}.
\]
With $\epsilon$ fixed, the left-hand-side is a continuous function of $s$,
denoted as $\gamma(s)$. When $s = s^{-}$, $s_{i} < \mu$ for all $-2K \leq i
\leq2K$. So,
\[
\gamma(s^{-}) < \frac{\sum_{k = -K}^{K} A_{k}(s, \epsilon) \cdot\frac{\Pr[(h,
b)]}{\Pr[(l, a)]}}{\sum_{k = -K}^{K} A_{k}(s, \epsilon)} = \frac{\Pr[(h,
b)]}{\Pr[(l, a)]}.
\]

Similarly, we have
\[
\gamma(s^{+}) > \frac{\Pr[(h, b)]}{\Pr[(l, a)]}.
\]

Therefore, there must exist $s \in(s^{-}, s^{+})$ such that
\[
\gamma(s) = \frac{\Pr[(h, b)]}{\Pr[(l, a)]},
\]
which implies that \eqref{eq:prob} is satisfied, hence completing the proof.
\end{proof}

\begin{proof}
[Proof of Lemma~\ref{lem:efficient}]We now argue that for sufficiently small
$\epsilon$ and properly selected $K$, the information structure defined by
Lemma~\ref{lem:prob} has $x_{K}, x^{\prime}_{-K} = o(1)$.

Consider the following Taylor expansion near $\mu$,
\[
\frac{h - z}{z - l} = \frac{h - \mu}{\mu- l} + O(z - \mu) = \frac{h - b}{a- l}
+ O(z - \mu).
\]

Similarly,
\[
\frac{a- z}{z - b} = \frac{a- \mu}{\mu- b} + O(z - \mu) = \frac{a - l}{h - b}
+ O(z - \mu).
\]

Therefore,
\[
A_{k}(s, \epsilon) = 1 + O(\Delta s), \quad A^{\prime}_{k}(s, \epsilon) = 1 +
O(\Delta s),
\]
where $\Delta s = \max\{|s_{-2K} - \mu|, |s_{2K} - \mu|\}$.

Then, we have
\[
x_{K} = x_{0} \cdot A_{K}(s, \epsilon) = \Pr[(l, a)] \cdot\frac{A_{K}(s,
\epsilon)}{\sum_{-K}^{K} A_{k}(s, \epsilon)} = \Pr[(l, a)] \cdot\frac{1}{2K +
1} + O(\Delta s),
\]
and similarly,
\[
x^{\prime}_{-K} = x^{\prime}_{0} \cdot A^{\prime}_{-K}(s, \epsilon) = \Pr[(h,
b)] \cdot\frac{A^{\prime}_{-K}(s, \epsilon)}{\sum_{-K}^{K} A^{\prime}_{k}(s,
\epsilon)} = \Pr[(h, b)] \cdot\frac{1}{2K + 1} + O(\Delta s).
\]

In other words, as long as $K \rightarrow\infty$ and $\Delta s \rightarrow0$,
$x_{K}, x^{\prime}_{-K} \rightarrow0$.

In particular, as $\epsilon\rightarrow0$, the following $K$ satisfy the
desired property:
\[
K = \frac{\sqrt{\epsilon}}4 \log_{1+\epsilon} \min\{h^{*} / \mu, \mu/ l^{*}\}
= O(\sqrt{\epsilon} / \ln(1 + \epsilon)) = O(1/\sqrt{\epsilon}) \rightarrow
\infty,
\]
in the meanwhile,
\[
\Delta s \leq\mu\cdot(1 + \epsilon)^{4K} - \mu\leq\mu(\min\{h^{*} / \mu, \mu/
l^{*}\}^{\sqrt{\epsilon}} - 1) = O(\sqrt{\epsilon}) \rightarrow0.
\qedhere\]
\end{proof}

\section{Additional Results and Proofs for Asymmetric Setting}\label{app:moderate}

In Section~\ref{sec:uni-inc} we provide a proof of Proposition~\ref{pro:uni-inc} for the uniform winner and incongruent loser case, which shows that the optimal information structure must be an interior information structure. In Section~\ref{sec:uni-con} and Section~\ref{sec:variable} we state additional propositions for (i) the uniform winner and congruent loser case and (ii) the variable winner case, which establish that the optimal information structure must be moderate. Together, these results show Theorem~\ref{thm:moderate}, whose proof appears in Section~\ref{sec:proof-moderate}.


\subsection{Uniform Winner, Incongruent Loser}\label{sec:uni-inc}


\begin{proof}
[Proof of Proposition~\ref{pro:uni-inc}]
Note that in this case $r_1 \geq r'_1 \geq r'_2 \geq r_2$ and hence also $\mu_1 \geq \mu_2$. Thus, the revenue from the extremal information structures is given by $\mu_2 = pr_2 + (1-p)r'_2$. Consequently, no-disclosure and full-disclosure yield the same expected revenue. Moreover, unbundling the winner and bundling the loser yield the same
revenue as no-disclosure or full-disclosure:%
\begin{equation*}
r_{1}p\left( \frac{r_{2}p+r_{2}^{\prime }\left( 1-p\right) }{r_{1}}\right)
+r_{1}^{\prime }\left( 1-p\right) \left( \frac{r_{2}p+r_{2}^{\prime }\left(
1-p\right) }{r_{1}^{\prime }}\right) =pr_{2}+\left( 1-p\right) r_{2}^{\prime
}
\end{equation*}
We then consider bundling the winner and unbundling the loser: 
\begin{equation*}
r_{1}p\left( \frac{r_{2}}{r_{1}p+r_{1}^{\prime }\left( 1-p\right) }\right)
+r_{1}^{\prime }\left( 1-p\right) \left( \frac{r_{2}^{\prime }}{%
r_{1}p+r_{1}^{\prime }\left( 1-p\right) }\right) 
\end{equation*}%
This yields a strictly lower revenue than no-disclosure with an incongruent
loser:%
\begin{eqnarray*}
&&r_{1}p\left( \frac{r_{2}}{r_{1}p+r_{1}^{\prime }\left( 1-p\right) }\right)
+r_{1}^{\prime }\left( 1-p\right) \left( \frac{r_{2}^{\prime }}{%
r_{1}p+r_{1}^{\prime }\left( 1-p\right) }\right) -\left(
pr_{2}+r_{2}^{\prime }\left( 1-p\right) \right)  \\
&=&\frac{p\left( 1-p\right) \left( r_{1}-r_{1}^{\prime }\right) \left(
r_{2}-r_{2}^{\prime }\right) }{pr_{1}+r_{1}^{\prime }\left( 1-p\right) }<0,
\end{eqnarray*}%
as the incongruent loser is defined by the property%
\begin{equation*}
\left( r_{1}-r_{1}^{\prime }\right) \left( r_{2}-r_{2}^{\prime }\right) <0%
\text{.}
\end{equation*}%
Thus, the exact opposite result to the congruent loser. Finally, unbundling
the winner and bundling the loser yields the same revenue as
no-disclosure or full-disclosure. 

It therefore suffices to show that we can construct an interior information
structure such that revenue improves in comparison to no-disclosure. We begin by partially
bundling the loser only. This is a single-agent optimization problem: 
\begin{equation}
\begin{tabular}{l|cccc}
\toprule
\backslashbox{$r$}{$s$} & $(r_{1},x)$ & $(r_{1},y)$ & $\left( r_{1}^{\prime },x\right) $ & $%
\left( r_{1}^{\prime },y\right) $ \\ \hline
$(r_{1},r_{2})$ & $p-q$ & $q$ & $0$ & $0$ \\ \hline
$(r_{1}^{\prime },r_{2}^{\prime })$ & $0$ & $0$ & $q^{\prime }$ & $p^{\prime
}-q^{\prime }$\\%
\bottomrule
\end{tabular}%
,\ \ \   \label{osp2}
\end{equation}%
where $p^{\prime }=1-p$. This information structure leads the following
revenue: 
\begin{equation}
r_{1}\left( \left( p-q\right) \frac{x}{r_{1}}+q\frac{y}{r_{1}}\right)
+r_{1}^{\prime }\left( q^{\prime }\frac{x}{r_{1}^{\prime }}+\left( p^{\prime
}-q^{\prime }\right) \frac{y}{r_{1}^{\prime }}\right)   \label{osp3}
\end{equation}%
where $r_{1}$ and $r_{1}^{\prime }$ cancel:%
\begin{equation*}
\left( \left( p-q\right) x+qy\right) +\left( q^{\prime }x+\left( p^{\prime
}-q^{\prime }\right) y\right) 
\end{equation*}%
and the calibration constraints are given by: 
\begin{equation*}
x=\frac{r_{2}\left( p-q\right) +r_{2}^{\prime }q^{\prime }}{\left(
p-q\right) +q^{\prime }}
\end{equation*}%
and 
\begin{equation*}
y=\frac{r_{2}q+r_{2}^{\prime }\left( p^{\prime }-q^{\prime }\right) }{%
q+p^{\prime }-q^{\prime }}.
\end{equation*}%
Thus, we have revenue indifference (as expected): 
\begin{equation*}
\left( \left( p-q\right) x+qy\right) +\left( q^{\prime }x+\left( p^{\prime
}-q^{\prime }\right) y\right) =r_{2}p+p^{\prime }r_{2}^{\prime }.
\end{equation*}%
It follows that the introduction of $\left( q,q^{\prime }\right) $, subject
to the probability constraints, leaves revenue unchanged. Clearly, we can
always choose $(q,q^{\prime })$ so that 
\begin{equation}
y>x\text{.}  \label{ops4}
\end{equation}%
However, we now notice that if we consider the submatrix: 
\begin{equation*}
\begin{tabular}{c|c}
\toprule
$(r_{1},y)$ & $\left( r_{1}^{\prime },x\right) $ \\ \hline
$q$ & $0$ \\ 
$0$ & $q^{\prime }$\\%
\bottomrule
\end{tabular}%
\ 
\end{equation*}%
and $y>x$, then this submatrix represents an instance of the congruent
winner with weak competition case, and hence by Theorem~\ref{thm:uni-con-weak} 
it is optimal
to bundle the winner. Thus,
\begin{equation*}
\begin{tabular}{c|c}
\toprule
$(z,y)$ & $\left(z,x\right) $ \\\hline
$q$ & $0$ \\ 
$0$ & $q^{\prime }$\\%
\bottomrule
\end{tabular}%
\ 
\end{equation*}%
with%
\begin{equation*}
z=\frac{r_{1}q+r_{1}^{\prime }q^{\prime }}{q+q^{\prime }}.
\end{equation*}
Moreover, given that the benefit of bundling is increasing in the difference 
$y-x$, it follows that we have an interior solution, as a higher $y$ and a
lower $x$ mean that $q$ and $q^{\prime }$ have to become smaller.
\end{proof}

\subsection{Uniform Winner, Congruent Loser}\label{sec:uni-con}


\begin{proposition}[Uniform Winner, Congruent Loser]
\label{pro:uni-con}
With a uniform winner and a congruent loser,
partially bundling the winner 
\emph{and} the loser is revenue-improving relative to the extremal
information structures. Both partial unbundling of the loser:%
\begin{equation}
\begin{tabular}{l|c|c|c}
\toprule
\backslashbox{$r$}{$s$} & $\left( \mu _{1},\mu _{1}\right)$  & $\left( \mu _{1},\mu _{1}\right)$  & 
$\left( \mu _{1},r_{2}^{\prime }\right)$  \\ \hline
$(r_{1},r_{2})$ & $p$ & $0$ & $0$ \\ 
$(r_{1}^{\prime },r_{2}^{\prime })$ & $0$ & $1-p-q$ & $q$\\
\bottomrule%
\end{tabular}
\label{u1}
\end{equation}%
or partial bundling of the winner:%
\begin{equation}
\begin{tabular}{l|c|c|c}
\toprule
\backslashbox{$r$}{$s$} & $\left( r_{2},r_{2}\right)$  & $\left( r_{2},r_{2}^{\prime }\right)$  & 
$\left( r_{1}^{\prime },r_{2}^{\prime }\right)$  \\\hline 
$(r_{1},r_{2})$ & $p$ & $0$ & $0$ \\ 
$(r_{1}^{\prime },r_{2}^{\prime })$ & $0$ & $1-p-q$ & $q$\\%
\bottomrule
\end{tabular}
\label{u2}
\end{equation}%
for suitably chosen $q$ 
generates a strictly higher revenue than the extremal information structures.
\end{proposition}


\begin{proof}
[Proof of Proposition~\ref{pro:uni-con}]
We begin by observing that in the uniform winner and congruent loser case the revenue from no- and full-disclosure is equal to
\[
\mu_2 = pr_2 + (1-p)r'_2.
\]

Next, we argue that the information structure in  (\ref{u1}) yields a strict revenue improvement. Its revenue is given by:%
\begin{equation*}
r_{1}p+r_{1}^{\prime }(1-p-q)+r_{1}^{\prime }q\frac{r_{2}^{\prime }}{\mu _{1}%
},
\end{equation*}%
where the calibration constraint is given by: 
\begin{equation*}
\frac{r_{2}p+r_{2}^{\prime }(1-p-q)}{1-q}=r_{1}p+r_{1}^{\prime }(1-p)=\mu
_{1},
\end{equation*}%
where we remind the reader that 
\begin{equation*}
\mu _{1}=r_{1}p+r_{1}^{\prime }(1-p)
\end{equation*}%
and thus: 
\begin{equation*}
q=-\frac{\left( pr_{1}-pr_{2}-r_{1}^{\prime }\left( p-1\right)
+r_{2}^{\prime }\left( p-1\right) \right) }{r_{2}^{\prime
}-pr_{1}+r_{1}^{\prime }\left( p-1\right) }
\end{equation*}%
Hence, the revenue difference is%
\begin{eqnarray*}
&&r_{1}p+r_{1}^{\prime }(1-p-q)+qr_{1}^{\prime }\frac{r_{2}^{\prime }}{\mu
_{1}}-\left( r_{2}p+\left( 1-p\right) r_{2}^{\prime }\right)  \\
&=&p\frac{r_{1}-r_{1}^{\prime }}{\left( 1-p\right) r_{1}^{\prime }+pr_{1}}%
\left( \left( 1-p\right) \left( r_{1}^{\prime }-r_{2}^{\prime }\right)
+p\left( r_{1}-r_{2}\right) \right) >0,
\end{eqnarray*}%
where the inequality holds due to the uniform winner property. 

It remains to establish that the revenue of (\ref{u2}) is also greater than the revenue from the extremal information structures. The revenue of (\ref{u2}) is:%
\begin{equation*}
r_{1}p+r_{1}^{\prime }(1-p-q)\frac{r_{2}^{\prime }}{r_{2}}+r_{1}^{\prime }q%
\frac{r_{2}^{\prime }}{r_{1}^{\prime }},
\end{equation*}%
where the calibration constraint is given by: 
\begin{equation*}
\frac{r_{1}p+r_{1}^{\prime }(1-p-q)}{1-q}=r_{2},
\end{equation*}%
and thus: 
\begin{equation*}
q=-\frac{1}{r_{1}^{\prime }-r_{2}}\left( r_{2}-pr_{1}+r_{1}^{\prime }\left(
p-1\right) \right) .
\end{equation*}%
The revenue difference is therefore 
\begin{eqnarray*}
&&r_{1}p+r_{1}^{\prime }(1-p-q)\frac{r_{2}^{\prime }}{r_{2}}+r_{1}^{\prime }q%
\frac{r_{2}^{\prime }}{r_{1}^{\prime }}-\left( r_{2}p+\left( 1-p\right)
r_{2}^{\prime }\right)  \\
&=&\frac{p}{r_{2}}\left( r_{1}-r_{2}\right) \left( r_{2}-r_{2}^{\prime
}\right) >0,
\end{eqnarray*}%
where the inequality holds by the congruent loser property. 
\end{proof}






\subsection{Variable Winner}
\label{sec:variable}


\begin{proposition}[Variable Winner]
\label{pro:variable}
With a variable winner, 
a moderate information structure 
improves on the extremal information structures.
\end{proposition}


\begin{proof}
Note that the revenue from full-disclosure is $p r_1 \frac{r'_1}{r_1} + (1-p) r'_2 \frac{r_2}{r'_2} = p r'_1 + (1-p)r_2$, which is no more than $\mu_2$ because $r'_1 \leq r'_2$, but also no more than $\mu_1$ because $r_2 \leq r_1$. So, the revenue from full-disclosure is at most $ \min\{\mu_1,\mu_2\}$. The revenue from no-disclosure, on the other hand, is $\mu_2$ if $\mu_1 \geq \mu_2$, and $\mu_1$ otherwise. It is therefore equal to $\min\{\mu_1,\mu_2\}$.

We construct a moderate information structure that strictly improves upon this.
Without loss of
generality, we assume that $\mu _{1}>\mu _{2}$. We consider the following
candidate information structure: 
\begin{equation}
\begin{tabular}{l|ccc}
\toprule
\backslashbox{$r$}{$s$}& $(r_{1},r_{2})$ & $(r_{1}^{\prime },r_{2}^{\prime })$ & $\left(
r_{1}^{\prime \prime },r_{2}^{\prime \prime }\right)$  \\ \hline
$(r_{1},r_{2})$ & $\epsilon \lambda$  & $0$ & $1-\epsilon$  \\ 
$(r_{1}^{\prime },r_{2}^{\prime })$ & $0$ & $\epsilon \left( 1-\lambda \right)$ 
& $1-\epsilon$\\
\bottomrule
\end{tabular}%
,  \label{t}
\end{equation}%
the construction of which we now detail. By the hypothesis of a variable
winner, namely 
\begin{equation}
r_{1}>r_{2},\ \ r_{1}^{\prime }<r_{2}^{\prime },  \label{vari}
\end{equation}
we can find a unique convex combination $\lambda $ and $1-\lambda $ of the
click-through rate configurations $(r_{1},r_{2})$ and $(r_{1}^{\prime
},r_{2}^{\prime })$, such that the expected click-through rate of both
bidders is the same. Thus, in expectation they have the same expected
click-through rate, which we call $\hat{r}$: 
\begin{equation}
\lambda r_{1}+\left( 1-\lambda \right) r_{1}^{\prime }=\lambda r_{2}+\left(
1-\lambda \right) r_{2}^{\prime }= \hat{r}.  \label{s}
\end{equation}
The remaining probability is collected in a single calibrated signal that
yields the posterior expectation of the click-through rates:%
\begin{equation*}
r_{i}^{\prime \prime }=\frac{\left( p-\epsilon \lambda \right)
r_{i}+\left( 1-p-\epsilon \left( 1-\lambda \right) \right) r_{i}^{\prime }%
}{\left( p-\epsilon \lambda \right) +\left( 1-p-\epsilon \left(
1-\lambda \right) \right) }.
\end{equation*}%
This yields the above information structure (\ref{t}). Now, we can always
chose $\epsilon >0$ sufficiently small so that $r_{1}^{\prime \prime
}\geq r_{2}^{\prime \prime }.$

If we restrict attention to the subset of signals $%
(r_{1},r_{2}),(r_{1}^{\prime },r_{2}^{\prime })$, then it is as if we have a
bidding game with equal means of the CTRs by construction of (\ref{s}). By
Lemma 4.5, we know that on this subset of signals we can extract the full
surplus generated by the click-through rates. In the remaining profile, we
run the bidding game as if we have a zero-disclosure policy. With this, we
know that the resulting revenue is given by 
\begin{equation}
\epsilon \lambda r_{1}+\left( 1-\epsilon \right) \lambda r_{2}^{\prime
}+\left( 1-\epsilon \right) \frac{\left( p-\epsilon \lambda \right)
r_{2}+\left( 1-p-\epsilon \left( 1-\lambda \right) \right) r_{2}^{\prime }%
}{\left( p-\epsilon \lambda \right) +\left( 1-p-\epsilon \left(
1-\lambda \right) \right) }.  \label{s1}
\end{equation}%
By contrast, in the no-disclosure policy the revenue would be $\mu _{2}$,
which can be written as%
\begin{equation}
\epsilon \lambda r_{2}+\left( 1-\epsilon \right) \lambda r_{2}^{\prime
}+\left( 1-\epsilon \right) \frac{\left( p-\epsilon \lambda \right)
r_{2}+\left( 1-p-\epsilon \left( 1-\lambda \right) \right) r_{2}^{\prime }%
}{\left( p-\epsilon \lambda \right) +\left( 1-p-\epsilon \left(
1-\lambda \right) \right) },  \label{s2}
\end{equation}%
and by the hypothesis of the variable winner (see (\ref{vari})), $r_{1}>r_{2}$%
, and the above information structure (\ref{t}) strictly improves upon the
no-disclosure policy.
\end{proof}


\subsection{Proof of Theorem \ref{thm:moderate}}\label{sec:proof-moderate}

\begin{proof}[Proof of Theorem \ref{thm:moderate}]
The claim follows from combining Proposition~\ref{pro:uni-inc} with Proposition~\ref{pro:uni-con} and Proposition~\ref{pro:variable}.
\end{proof}

\end{document}